\documentclass[10pt]{article}

\usepackage{graphicx}
\usepackage{amsmath}
\usepackage{amsthm}

\usepackage{url}

\usepackage[usenames,dvipsnames]{xcolor} 
\usepackage{tikz}
\usepackage{pgfplots}
\usepackage{subcaption} 
\usetikzlibrary{arrows,shapes,snakes,backgrounds,calc,shadows,shapes.callouts,positioning} 
\tikzset{>=latex} 

\usepackage[outline]{contour}

\usepackage{algorithm}
\usepackage{algpseudocode}

\renewcommand*\Call[2]{\textproc{#1}(#2)}

\usepackage{array}
\newcolumntype{L}[1]{>{\raggedright\let\newline\\\arraybackslash\hspace{0pt}}m{#1}}
\newcolumntype{C}[1]{>{\centering\let\newline\\\arraybackslash\hspace{0pt}}m{#1}}
\newcolumntype{R}[1]{>{\raggedleft\let\newline\\\arraybackslash\hspace{0pt}}m{#1}} 

\algnewcommand{\algorithmicand}{\textbf{and}}
\algnewcommand{\algorithmicnot}{\textbf{not}}
\algnotext{EndFor}
\algnotext{EndWhile}
\algnotext{EndIf}

\DeclareMathOperator*{\argmin}{arg\,min}

\newtheorem{thm}{Theorem}
\newtheorem{prp}{Proposition}
\newtheorem{lem}[thm]{Lemma}
\newtheorem{cor}{Corollary}

\newtheorem{dfn}{Definition}

\title{A comprehensive theory of cascading via-paths and the reciprocal pointer chain method} 
\author{Brandon Smock, Joseph Wilson\thanks{Computer and Information Science and Engineering Department, E301 CSE Building, PO Box 116120, Gainesville, FL 32611.}}

\begin{document}
\maketitle

\begin{abstract}
In this paper, we consolidate and expand upon the current theory and potential applications of the set of $k$ best \emph{cascading via-paths} (CVPs) and the \emph{reciprocal pointer chain} (RPC) method for identifying them.
CVPs are a collection of up to $|V|$ paths between a source and a target node in a graph $G = (V,E)$, computed using two shortest path trees, that have distinctive properties relative to other path sets.
They have been shown to be particularly useful in geospatial applications, where they are an intuitive and efficient means for identifying a set of spatially diverse alternatives to the single shortest path between the source and target.
However, spatial diversity is not intrinsic to paths in a graph, and little theory has been developed outside of application to describe the nature of these paths and the RPC method in general.
Here we divorce the RPC method from its typical geospatial applications and develop a comprehensive theory of CVPs from an abstract graph-theoretic perspective.
Restricting ourselves to properties of the CVPs and of the entire set of $k$-best CVPs that can be computed in $O(|E| + |V| \log |V|)$, we are able to then propose, among other things, new and efficient approaches to problems such as generating a diverse set of paths and to computing the $k$ shortest loopless paths between two nodes in a graph.
We conclude by demonstrating the new theory in practice, first for a typical application of finding alternative routes in road networks and then for a novel application of identifying layer-boundaries in ground-penetrating radar (GPR) data.
It is our hope that by generalizing the RPC method, providing a sound theoretical foundation, and demonstrating novel uses, we are able to broaden its perceived applicability and stimulate new research in this area, both applied and theoretical.
\end{abstract}

\section{Introduction}\label{sec:intro}

A \emph{via-path} is the shortest path in a graph $G = (V,E)$ from a source node, $s$, to a target node, $t$, that passes through an intermediate node, $v$.
From the perspective of graph theory, the $k$-best via-paths form a convenient path \emph{set}, comparable to the $k$ shortest loopless paths \cite{yen1971finding} or $k$ shortest disjoint paths \cite{suurballe1974disjoint}, but with its own unique collection of properties.
From an applied perspective, via-paths have proved to be useful in a growing number of applications, such as finding alternate routes in road networks \cite{abraham2013alternative}, power line path planning \cite{li2013optimization}, animal dispersal route analysis in landscape ecology \cite{pinto2009beyond}, autonomous robot navigation \cite{shiller2004computing}, and layer-boundary identification in ground-penetrating radar (GPR) data \cite{smock2012reciprocal}.

The standard method for computing a set of via-paths is based on computing two shortest path trees, using a procedure such as the bi-directional Dijkstra's algorithm \cite{nicholson1966finding}.
The first (or, \emph{predecessor}) shortest path tree, $G_\beta$, encodes the shortest path $p_{s,v}$ from the source node, $s$, to each node $v \in V$.
The second (or, \emph{successor}) shortest path tree, $G_\phi$, encodes the shortest path $p_{v,t}$ from each node $v \in V$ to the target node, $t$.
Together, $p_{s,v}$ and $p_{v,t}$ combine to form the via-path $p_{s,v,t}$ for each node $v \in V$.

This two-tree method implicitly computes a single $s$--$t$ via-path for each node $v \in V$.
This is not necessarily the set of all $s$--$t$ via-paths, since multiple paths could tie for the best via-path for a given node.
Which single via-path is chosen from the set of tied candidates for a given node during tree computation affects the particular via-path computed for that node's descendent nodes in the tree.
Therefore, to distinguish the set of via-paths computed by this two-tree method from other possible sets, we refer to this set as the set of \emph{cascading} via-paths (CVPs).

A characteristic feature associated with each cascading via-path, which we denote $\tilde{p}_{s,v,t}$, in $G$ is its \emph{reciprocal pointer chain} (RPC).
A reciprocal pointer chain is a path $p_{v,u} \in G_\beta$ whose transpose $p_{u,v}$ is in $G_\phi$.
It can be shown that each CVP has a unique RPC associated with it---and even more interesting, it can be shown that the RPCs of any two CVPs are \emph{disjoint}.
In this paper, we spend considerable effort expounding further on the properties of reciprocal pointer chains, which turn out to be useful both for characterizing the properties of their associated cascading via-paths and as a feature in their own right.

Despite the fact that the two-tree method for computing CVPs has been known for some time, many of the properties of CVPs and RPCs from a theoretical standpoint have gone largely undocumented
This can be attributed to the fact that the work in this area up to this point has been primarily application-motivated.
Consequently, CVPs and the two-tree method have been independently discovered several times over the last 25 years.

In 1993, Lombard and Church \cite{lombard1993gateway} provided the first description of a via-path as a \emph{gateway shortest path} and described the two-tree method for computing what we refer to here as the set of CVPs.
In 2003, Fujita, Nakamura, and Shiller \cite{fujita2003dual} proposed the same two-tree method for computing CVPs, under the name \emph{dual Dijkstra search}, and filed for a patent \cite{nakamura2003dual} on the method that same year.
In 2006, the two-tree method for computing CVPs was proposed again in a proprietary system by Cambridge Vehicle Information Technology (CAMVIT), for application to finding alternative routes in road networks \cite{camvit2006choice}, introducing the concept of a \emph{plateau} and describing a chain of nodes that in this work we refer to as a reciprocal pointer chain.
This work led to more patents \cite{jones2012method,jones2012method2} on the two-tree method being filed in 2007.
Smock and Wilson \cite{smock2012reciprocal,smock2012efficient} proposed the two-tree method in 2012 in the context of trellis graphs and the Viterbi algorithm, although instead of using it to compute CVPs, they introduced the idea of computing just the RPCs as the end product.

This continual re-discovery has led to the lack of a standardized terminology.
Furthermore, all of these applications, including the original work by Lombard and Church, assume the nodes in $G$ have fixed coordinates, and many assume additional restrictions on the form of the edges between the nodes.
As a result, CVPs and the RPC method have never been formalized or abstracted in a way suitable for application to all graphs.
Additionally, previous claims about algorithms for computing via-paths can be shown to not hold true when these assumptions are lifted.

In this paper, we address these issues and establish for the first time a comprehensive, standardized framework and general theory of cascading via-paths and reciprocal pointer chains, placing particular emphasis on the computation and properties of the set of $k$-best CVPs.
In Section \ref{sec:method}, we review and formalize the procedure for computing the $k$-best CVPs.
By approaching the method from a graph-theoretic perspective, we can begin to sort out algorithmic inconsistencies that did not arise in earlier work.
In Section \ref{sec:theory}, we consolidate and present new theory of CVPs and RPCs.
Among our mathematical contributions, we describe RPCs in detail for the first time and show that numerous properties of CVPs derive from the properties of their RPCs, including concrete notions of path local optimality and path set diversity.
Among our computational contributions, we prove theorems describing the relationship between CVPs and the $k$ shortest paths in a graph that suggest a potentially significantly faster way to compute the $k$ shortest loopless paths using Yen's algorithm \cite{yen1971finding}.

In Section \ref{sec:results}, we present two applications that demonstrate the usefulness and versatility of the RPC method when considered from a graph-theoretic perspective.
The first application is a standard use of the method to produce a set of candidate paths, but we produce the candidates in a more principled way founded on the new theory established in this paper.
The second application is a novel use of the RPC method that leverages the properties of RPCs specifically, independent of their corresponding CVPs, to identify linear, physical features in ground-penetrating radar (GPR) images.
In demonstrating these uses and showing how they arise naturally from theoretical considerations, we aim to further establish the generality of the approach and its potential use in an even broader class of applications.
Finally, in Section \ref{sec:conclusion}, we conclude with a review of our major results and their implications.

\section{Background and Method}\label{sec:method}

In this section, we standardize and describe in detail a method to compute the set of $k$-best cascading via-paths (CVPs), given a graph $G = (V,E)$, source node $s$, target node $t$, and positive integer $k$.
We denote the set of all distinct CVPs as $\Lambda$, an ordered/ranked set of CVPs as $\hat{\Lambda}$, and an ordered subset consisting of the $k$ best CVPs as $\hat{\Lambda}^{(k)}$.
We begin with a summary of the procedure, which can be broken into four primary stages.
These stages are depicted in Figure \ref{fig:cvp_computation}, with each transition between graphs corresponding to a single stage.
\subsection{Method summary}\label{sec:summary}
\begin{enumerate}
	\item \emph{Implicit via-path computation}: Given $G$, $s$, and $t$, compute a predecessor shortest path tree $G_\beta$ and a successor shortest path tree $G_\phi$.
	\item \emph{Cascading via-path enumeration (graph partitioning)}: Given $G_\beta$ and $G_\phi$, enumerate the elements of $\Lambda$ by grouping together nodes $v \in V$ with the same CVP. This produces $\Pi$, a set of disjoint partitions of $V$, where $|\Pi| = |\Lambda|$.
	\item \emph{Via-path ranking and sub-selection}: Given $\Pi$ and a positive integer $k$, sort the paths/partitions by some measure of goodness, producing the sequence $\hat{\Pi}$, and select only the top $k$ paths, which are jointly denoted $\hat{\Pi}^{(k)}$.
	\item \emph{Via-path extraction}: Given $G_\beta$, $G_\phi$, and $\hat{\Pi}^{(k)}$, for each CVP, extract its explicit sequence of nodes from its implicit representation within the shortest path trees using a depth-first search, producing the set of explicit $k$-best CVPs, $\hat{\Lambda}^{(k)}$.
\end{enumerate}

\begin{figure}
	 \centering
	    \begin{subfigure}[]{0.19\textwidth}
	    \centering
	    \resizebox{\textwidth}{!}{
        \begin{tikzpicture} [
            basic/.style={
                circle,
                minimum size=5mm,
                draw=black,
                font=\tiny,
                thick
            },
            source_target/.style={
                minimum size=7mm,
                very thick,
                font=\small,
            },
            basic_edge/.style={
                ->,
                draw=black,
                very thick,
            },
            edge_weight/.style={
                font=\footnotesize
            }]
                \node (a) [basic,source_target] at (0,2.5) {s};
                \node (b) [basic] at (1,0.75) {};
                \node (c) [basic] at (1.5,4) {};
                \node (d) [basic] at (2,2.25) {};
                \node (e) [basic] at (2.75,0.25) {};
                \node (f) [basic] at (3.25,4.5) {};
                \node (g) [basic] at (3.75,3) {};
                \node (h) [basic] at (4.25,1.5) {};
                \node (i) [basic] at (5,4) {};
                \node (j) [basic] at (5.75,2.25) {};
                \node (k) [basic] at (6,0.5) {};
                \node (l) [basic] at (7.25,3.5) {};
                \node (m) [basic] at (7.5,1.5) {};
                \node (n) [basic,source_target] at (8.75,2.5) {t};
                
                \draw (9.75,2.5) edge [->,line width=5pt] (10.75,2.5); 
                
                \draw (a) edge [basic_edge] node[edge_weight,anchor=west]{3} (b);
                \draw (a) edge [basic_edge] node[edge_weight,anchor=south]{3} (c);
                \draw (a) edge [basic_edge] node[edge_weight,anchor=south]{2} (d);
                \draw (b) edge [basic_edge] node[edge_weight,anchor=south]{2} (d);
                \draw (b) edge [basic_edge] node[edge_weight,anchor=south]{3} (e);
                \draw (c) edge [basic_edge] node[edge_weight,anchor=east]{2} (d);
                \draw (c) edge [basic_edge] node[edge_weight,anchor=south]{2} (f);
                \draw (c) edge [basic_edge] node[edge_weight,anchor=south]{4} (g);
                \draw (d) edge [basic_edge] node[edge_weight,anchor=south]{4} (h);
                \draw (d) edge [basic_edge] node[edge_weight,anchor=south]{3} (g);
                \draw (d) edge [basic_edge] node[edge_weight,anchor=west]{3} (e);
                \draw (e) edge [basic_edge] node[edge_weight,anchor=south]{2} (h);
                \draw (e) edge [basic_edge] node[edge_weight,anchor=south]{2} (k);
                \draw (f) edge [basic_edge] node[edge_weight,anchor=west]{2} (g);
                \draw (f) edge [basic_edge] node[edge_weight,anchor=south]{2} (i);
                \draw (g) edge [basic_edge] node[edge_weight,anchor=south]{3} (i);
                \draw (g) edge [basic_edge] node[edge_weight,anchor=west]{2} (h);
                \draw (g) edge [basic_edge] node[edge_weight,anchor=south]{2} (j);
                \draw (h) edge [basic_edge] node[edge_weight,anchor=south]{1} (j);
                \draw (h) edge [basic_edge] node[edge_weight,anchor=south]{3} (k);
                \draw (i) edge [basic_edge] node[edge_weight,anchor=west]{3} (j);
                \draw (i) edge [basic_edge] node[edge_weight,anchor=south]{3} (l);
                \draw (j) edge [basic_edge] node[edge_weight,anchor=west]{2} (k);
                \draw (j) edge [basic_edge] node[edge_weight,anchor=south]{3} (l);
                \draw (j) edge [basic_edge] node[edge_weight,anchor=south]{3} (m);
                \draw (j) edge [basic_edge] node[edge_weight,anchor=south]{4} (n);
                \draw (k) edge [basic_edge] node[edge_weight,anchor=south]{2} (m);
                \draw (l) edge [basic_edge] node[edge_weight,anchor=south]{3} (n);
                \draw (m) edge [basic_edge] node[edge_weight,anchor=north]{3} (n);
        \end{tikzpicture}
	}
		 \label{fig:method_part1}
	    \end{subfigure}
	    \begin{subfigure}[]{0.19\textwidth}
              \centering
             \resizebox{\textwidth}{!}{
        \begin{tikzpicture} [
            basic/.style={
                circle,
                minimum size=5mm,
                draw=black,
                font=\tiny,
                thick
            },
            source_target/.style={
                minimum size=7mm,
                very thick,
                font=\small,
            },
            basic_edge/.style={
                ->,
                draw=black,
                very thick,
            },
            reciprocal_edge/.style={
                ->,
                bend left=10,
                draw=black,
                very thick,
            },
            back_pointer/.style={
                densely dashed
            },
            edge_weight/.style={
                font=\footnotesize
            }]
                \node (a) [basic,source_target] at (0,2.5) {s};
                \node (b) [basic,fill=black!8] at (1,0.75) {12};
                \node (c) [basic,fill=black!16] at (1.5,4) {13};
                \node (d) [basic] at (2,2.25) {11};
                \node (e) [basic,fill=black!8] at (2.75,0.25) {12};
                \node (f) [basic,fill=black!16] at (3.25,4.5) {13};
                \node (g) [basic] at (3.75,3) {11};
                \node (h) [basic] at (4.25,1.5) {11};
                \node (i) [basic,fill=black!16] at (5,4) {13};
                \node (j) [basic] at (5.75,2.25) {11};
                \node (k) [basic,fill=black!8] at (6,0.5) {12};
                \node (l) [basic,fill=black!16] at (7.25,3.5) {13};
                \node (m) [basic,fill=black!8] at (7.5,1.5) {12};
                \node (n) [basic,source_target] at (8.75,2.5) {t};
                
                 \draw (9.75,2.5) edge [->,line width=5pt] (10.75,2.5);
                
                \draw (b) edge [basic_edge,back_pointer] node[edge_weight,anchor=west]{3} (a);
                \draw (c) edge [basic_edge,back_pointer] node[edge_weight,anchor=south]{3} (a);
                \draw (d) edge [reciprocal_edge,back_pointer] node[edge_weight,anchor=north]{2} (a);
                \draw (f) edge [reciprocal_edge,back_pointer] node[edge_weight,anchor=north]{2} (c);
                \draw (h) edge [basic_edge,back_pointer] node[edge_weight,anchor=north]{4} (d);
                \draw (g) edge [reciprocal_edge,back_pointer] node[edge_weight,anchor=north]{3} (d);
                \draw (e) edge [basic_edge,back_pointer] node[edge_weight,anchor=west]{3} (d);
                \draw (i) edge [reciprocal_edge,back_pointer] node[edge_weight,anchor=north]{2} (f);
                \draw (j) edge [reciprocal_edge,back_pointer] node[edge_weight,anchor=north]{2} (g);
                \draw (k) edge [reciprocal_edge,back_pointer] node[edge_weight,anchor=north]{2} (e);
                \draw (l) edge [reciprocal_edge,back_pointer] node[edge_weight,anchor=north]{3} (i);
                \draw (m) edge [reciprocal_edge,back_pointer] node[edge_weight,anchor=north]{2} (k);
                \draw (n) edge [reciprocal_edge,back_pointer] node[edge_weight,anchor=north]{4} (j);
                
                \draw (a) edge [reciprocal_edge] node[edge_weight,anchor=south]{2} (d);
                \draw (b) edge [basic_edge] node[edge_weight,anchor=south]{3} (e);
                \draw (c) edge [reciprocal_edge] node[edge_weight,anchor=south]{2} (f);
                \draw (d) edge [reciprocal_edge] node[edge_weight,anchor=south]{3} (g);
                \draw (e) edge [reciprocal_edge] node[edge_weight,anchor=south]{2} (k);
                \draw (f) edge [reciprocal_edge] node[edge_weight,anchor=south]{2} (i);                
                \draw (g) edge [reciprocal_edge] node[edge_weight,anchor=south]{2} (j);
                \draw (h) edge [basic_edge] node[edge_weight,anchor=north]{1} (j);                
                \draw (i) edge [reciprocal_edge] node[edge_weight,anchor=south]{3} (l);
                \draw (j) edge [reciprocal_edge] node[edge_weight,anchor=south]{4} (n);
                \draw (k) edge [reciprocal_edge] node[edge_weight,anchor=south]{2} (m);
                \draw (l) edge [basic_edge] node[edge_weight,anchor=south]{3} (n);
                \draw (m) edge [basic_edge] node[edge_weight,anchor=north]{3} (n);
        \end{tikzpicture}
        }
        \label{fig:method_part2}
	      
	    \end{subfigure}
	    \begin{subfigure}[]{0.19\textwidth}
              \centering
             \resizebox{\textwidth}{!}{
        \begin{tikzpicture} [
            basic/.style={
                circle,
                minimum size=5mm,
                draw=black,
                font=\tiny,
                thick
            },
            source_target/.style={
                minimum size=7mm,
                very thick,
                font=\small,
            },
            basic_edge/.style={
                ->,
                draw=black,
                very thick,
            },
            reciprocal_edge/.style={
                ->,
                bend left=10,
                draw=black,
                very thick,
            },
            back_pointer/.style={
                densely dashed
            },
            edge_weight/.style={
                font=\footnotesize
            }]
                \node (a) [basic,source_target] at (0,2.5) {s};
                \node (b) [basic,fill=black!8] at (1,0.75) {12};
                \node (c) [basic,fill=black!16] at (1.5,4) {13};
                \node (d) [basic] at (2,2.25) {11};
                \node (e) [basic,fill=black!8] at (2.75,0.25) {12};
                \node (f) [basic,fill=black!16] at (3.25,4.5) {13};
                \node (g) [basic] at (3.75,3) {11};
                \node (h) [basic] at (4.25,1.5) {11};
                \node (i) [basic,fill=black!16] at (5,4) {13};
                \node (j) [basic] at (5.75,2.25) {11};
                \node (k) [basic,fill=black!8] at (6,0.5) {12};
                \node (l) [basic,fill=black!16] at (7.25,3.5) {13};
                \node (m) [basic,fill=black!8] at (7.5,1.5) {12};
                \node (n) [basic,source_target] at (8.75,2.5) {t};
                
                 \draw (9.75,2.5) edge [->,line width=5pt] (10.75,2.5);
        
                \draw (d) edge [reciprocal_edge,back_pointer] node[edge_weight,anchor=north]{} (a);
                \draw (f) edge [reciprocal_edge,back_pointer] node[edge_weight,anchor=north]{} (c);
                \draw (g) edge [reciprocal_edge,back_pointer] node[edge_weight,anchor=north]{} (d);
                \draw (i) edge [reciprocal_edge,back_pointer] node[edge_weight,anchor=north]{} (f);
                \draw (j) edge [reciprocal_edge,back_pointer] node[edge_weight,anchor=north]{} (g);
                \draw (k) edge [reciprocal_edge,back_pointer] node[edge_weight,anchor=north]{} (e);
                \draw (l) edge [reciprocal_edge,back_pointer] node[edge_weight,anchor=north]{} (i);
                \draw (m) edge [reciprocal_edge,back_pointer] node[edge_weight,anchor=north]{} (k);
                \draw (n) edge [reciprocal_edge,back_pointer] node[edge_weight,anchor=north]{} (j);
                
                \draw (a) edge [reciprocal_edge] node[edge_weight,anchor=south]{} (d);
                \draw (c) edge [reciprocal_edge] node[edge_weight,anchor=south]{} (f);
                \draw (d) edge [reciprocal_edge] node[edge_weight,anchor=south]{} (g);
                \draw (e) edge [reciprocal_edge] node[edge_weight,anchor=south]{} (k);
                \draw (f) edge [reciprocal_edge] node[edge_weight,anchor=south]{} (i);                
                \draw (g) edge [reciprocal_edge] node[edge_weight,anchor=south]{} (j);           
                \draw (i) edge [reciprocal_edge] node[edge_weight,anchor=south]{} (l);
                \draw (j) edge [reciprocal_edge] node[edge_weight,anchor=south]{} (n);
                \draw (k) edge [reciprocal_edge] node[edge_weight,anchor=south]{} (m);
                
                \draw[thick] plot [smooth cycle] coordinates {(2,0.25) (2.75,-0.25) (6.1,-0.15) (7.7,0.8) (8,1.85) (7.5,2) (6,1) (2.75,0.75)};
                \draw[thick] plot [smooth cycle] coordinates {(0.5,0.75) (0.6,0.35) (1,0.25) (1.4,0.35) (1.5,0.75) (1.4,1.15) (1,1.25) (0.6,1.15)};
                \draw[thick] plot [smooth cycle] coordinates {(0.8,3.75) (1.5,3.5) (3.25,4) (5,3.5) (7.25,3) (7.75,3.4) (7.25,4) (5,4.5) (3.25,5) (1.3,4.5)};
                \draw[thick] plot [smooth cycle] coordinates {(-0.5,2.5) (0,2) (2,1.75) (3.75,2.5) (5.75,1.75) (8.75,2) (9.25,2.5) (8.75,3) (5.75,2.75) (3.75,3.5) (2,2.75) (0,3)};
                \draw[thick] plot [smooth cycle] coordinates {(3.75,1.5) (3.85,1.1) (4.25,1) (4.65,1.1) (4.75,1.5) (4.65,1.9) (4.25,2) (3.85,1.9)}; 
        \end{tikzpicture}
        }
        \label{fig:method_part3}
	      
	    \end{subfigure}
	    \begin{subfigure}[]{0.19\textwidth}
              \centering
             \resizebox{\textwidth}{!}{
        \begin{tikzpicture} [
            basic/.style={
                circle,
                minimum size=5mm,
                draw=black,
                font=\tiny,
                thick
            },
            grayed/.style={
                circle,
                minimum size=5mm,
                draw=black!50,
                font=\tiny,
                thick
            },
            source_target/.style={
                minimum size=7mm,
                very thick,
                font=\small,
            },
            basic_edge/.style={
                ->,
                draw=black,
                very thick,
            },
            reciprocal_edge/.style={
                ->,
                bend left=10,
                draw=black,
                very thick,
            },
            back_pointer/.style={
                densely dashed
            },
            edge_weight/.style={
                font=\footnotesize
            }]
                \node (a) [basic,source_target] at (0,2.5) {s};
                \node (c) [basic,fill=black!16] at (1.5,4) {13};
                \node (d) [basic] at (2,2.25) {11};
                \node (e) [basic,fill=black!8] at (2.75,0.25) {12};
                \node (f) [basic,fill=black!16] at (3.25,4.5) {13};
                \node (g) [basic] at (3.75,3) {11};
                \node (i) [basic,fill=black!16] at (5,4) {13};
                \node (j) [basic] at (5.75,2.25) {11};
                \node (k) [basic,fill=black!8] at (6,0.5) {12};
                \node (l) [basic,fill=black!16] at (7.25,3.5) {13};
                \node (m) [basic,fill=black!8] at (7.5,1.5) {12};
                \node (n) [basic,source_target] at (8.75,2.5) {t};
                
                 \draw (9.75,2.5) edge [->,line width=5pt] (10.75,2.5);
                
                \draw (d) edge [reciprocal_edge,back_pointer] node[edge_weight,anchor=north]{} (a);
                \draw (f) edge [reciprocal_edge,back_pointer] node[edge_weight,anchor=north]{} (c);
                \draw (g) edge [reciprocal_edge,back_pointer] node[edge_weight,anchor=north]{} (d);
                \draw (i) edge [reciprocal_edge,back_pointer] node[edge_weight,anchor=north]{} (f);
                \draw (j) edge [reciprocal_edge,back_pointer] node[edge_weight,anchor=north]{} (g);
                \draw (k) edge [reciprocal_edge,back_pointer] node[edge_weight,anchor=north]{} (e);
                \draw (l) edge [reciprocal_edge,back_pointer] node[edge_weight,anchor=north]{} (i);
                \draw (m) edge [reciprocal_edge,back_pointer] node[edge_weight,anchor=north]{} (k);
                \draw (n) edge [reciprocal_edge,back_pointer] node[edge_weight,anchor=north]{} (j);
                
                \draw (a) edge [reciprocal_edge] node[edge_weight,anchor=south]{} (d);
                \draw (c) edge [reciprocal_edge] node[edge_weight,anchor=south]{} (f);
                \draw (d) edge [reciprocal_edge] node[edge_weight,anchor=south]{} (g);
                \draw (e) edge [reciprocal_edge] node[edge_weight,anchor=south]{} (k);
                \draw (f) edge [reciprocal_edge] node[edge_weight,anchor=south]{} (i);                
                \draw (g) edge [reciprocal_edge] node[edge_weight,anchor=south]{} (j);           
                \draw (i) edge [reciprocal_edge] node[edge_weight,anchor=south]{} (l);
                \draw (j) edge [reciprocal_edge] node[edge_weight,anchor=south]{} (n);
                \draw (k) edge [reciprocal_edge] node[edge_weight,anchor=south]{} (m);
                
                \draw[thick] plot [smooth cycle] coordinates {(2,0.25) (2.75,-0.25) (6.1,-0.15) (7.7,0.8) (8,1.85) (7.5,2) (6,1) (2.75,0.75)};
                \draw[thick] plot [smooth cycle] coordinates {(0.8,3.75) (1.5,3.5) (3.25,4) (5,3.5) (7.25,3) (7.75,3.4) (7.25,4) (5,4.5) (3.25,5) (1.3,4.5)};
                \draw[thick] plot [smooth cycle] coordinates {(-0.5,2.5) (0,2) (2,1.75) (3.75,2.5) (5.75,1.75) (8.75,2) (9.25,2.5) (8.75,3) (5.75,2.75) (3.75,3.5) (2,2.75) (0,3)};
        \end{tikzpicture}
        }
        \label{fig:method_part4}
	      
	    \end{subfigure}
	    \begin{subfigure}[]{0.17\textwidth}
              \centering
             \resizebox{\textwidth}{!}{
        \begin{tikzpicture} [
            basic/.style={
                circle,
                minimum size=5mm,
                draw=black,
                font=\tiny,
                thick
            },
            source_target/.style={
                minimum size=7mm,
                very thick,
                font=\small,
            },
             gray_edge/.style={
                ->,
                draw=black!16,
                thin,
                dashed,
            },
            via_path1/.style={
                ->,
                draw=black,
                ultra thick,
            },
            via_path2/.style={
                ->,
                draw=black,
                ultra thick,
                densely dashed,
            },
            via_path3/.style={
                ->,
                draw=black,
                ultra thick,
                densely dotted,
            },
            edge_weight/.style={
                font=\footnotesize
            }]
                \node (a) [basic,source_target] at (0,2.5) {s};
                \node (b) [basic] at (1,0.75) {};
                \node (c) [basic] at (1.5,4) {};
                \node (d) [basic] at (2,2.25) {};
                \node (e) [basic] at (2.75,0.25) {};
                \node (f) [basic] at (3.25,4.5) {};
                \node (g) [basic] at (3.75,3) {};
                \node (h) [basic] at (4.25,1.5) {};
                \node (i) [basic] at (5,4) {};
                \node (j) [basic] at (5.75,2.25) {};
                \node (k) [basic] at (6,0.5) {};
                \node (l) [basic] at (7.25,3.5) {};
                \node (m) [basic] at (7.5,1.5) {};
                \node (n) [basic,source_target] at (8.75,2.5) {t};
                
                \draw (a) edge [via_path1] node[edge_weight,anchor=south]{} (d);
                \draw (d) edge [via_path1] node[edge_weight,anchor=south]{} (g);
                \draw (g) edge [via_path1] node[edge_weight,anchor=south]{} (j);
                \draw (j) edge [via_path1] node[edge_weight,anchor=south]{} (n);
 
 		\draw (a) edge [via_path2,bend right=15] node[edge_weight,anchor=south]{} (d);
                \draw (d) edge [via_path2] node[edge_weight,anchor=west]{} (e);               
                \draw (e) edge [via_path2] node[edge_weight,anchor=south]{} (k);
	        \draw (k) edge [via_path2] node[edge_weight,anchor=south]{} (m);
                \draw (m) edge [via_path2] node[edge_weight,anchor=north]{} (n);
                
                \draw (a) edge [via_path3] node[edge_weight,anchor=south]{} (c);
                \draw (c) edge [via_path3] node[edge_weight,anchor=south]{} (f);
                \draw (f) edge [via_path3] node[edge_weight,anchor=south]{} (i);
                \draw (i) edge [via_path3] node[edge_weight,anchor=south]{} (l); 
                \draw (l) edge [via_path3] node[edge_weight,anchor=south]{} (n);
                
                \draw (a) edge [gray_edge] node[edge_weight,anchor=west]{} (b);
                \draw (b) edge [gray_edge] node[edge_weight,anchor=south]{} (d);
                \draw (b) edge [gray_edge] node[edge_weight,anchor=south]{} (e);
                \draw (c) edge [gray_edge] node[edge_weight,anchor=east]{} (d);
                \draw (c) edge [gray_edge] node[edge_weight,anchor=south]{} (g);
                \draw (d) edge [gray_edge] node[edge_weight,anchor=south]{} (h);
                \draw (e) edge [gray_edge] node[edge_weight,anchor=south]{} (h);
                \draw (f) edge [gray_edge] node[edge_weight,anchor=west]{} (g);
                \draw (g) edge [gray_edge] node[edge_weight,anchor=south]{} (i);
                \draw (g) edge [gray_edge] node[edge_weight,anchor=west]{} (h);
                \draw (h) edge [gray_edge] node[edge_weight,anchor=south]{} (j);
                \draw (h) edge [gray_edge] node[edge_weight,anchor=south]{} (k);
                \draw (i) edge [gray_edge] node[edge_weight,anchor=west]{} (j);
                \draw (j) edge [gray_edge] node[edge_weight,anchor=west]{} (k);
                \draw (j) edge [gray_edge] node[edge_weight,anchor=south]{} (l);
                \draw (j) edge [gray_edge] node[edge_weight,anchor=south]{} (m);

        \end{tikzpicture}
        }
        \label{fig:method_part5}
	      
	\end{subfigure}
        \caption{An illustration of the four stages of the procedure for computing the set of $k$-best cascading via-paths.}
        \label{fig:cvp_computation}	 
\end{figure}
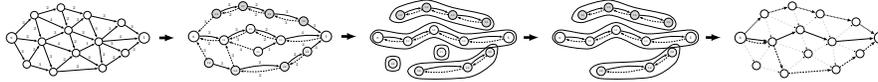

To fully describe the procedure, we now describe each of the stages in more detail.

\subsection{Implicit via-path computation}\label{sec:stage1} 

Let $G = (V,E)$ be a weighted, directed graph with vertex set $V$, edge set $E$, and non-negative edge weights.
Let $(u,v) \in E$ represent an edge from vertex $u$ to vertex $v$, and let $c(u,v)$ be the weight, or cost, associated with edge $(u,v)$, where we impose the restriction that $c(u,v) \ge 0$.
The first stage of the procedure for computing the $k$-best cascading via-paths is to compute two shortest path trees (SPT) from $G$---a \emph{predecessor} SPT, $G_\beta$, that encodes the shortest path from a source node $s$ to every node $v \in V$, and a \emph{successor} SPT, $G_\phi$, that encodes the shortest path from every node $v \in V$ to a target node $t$.

\subsubsection{Shortest path tree}\label{sec:shortestpathtree} 

\begin{figure}
    \centering
    \resizebox{0.4\textwidth}{!}{
        \begin{tikzpicture} [
            basic/.style={
                circle,
                minimum size=5mm,
                draw=black,
                font=\tiny,
                thick
            },
            source_target/.style={
                minimum size=7mm,
                very thick,
                font=\small,
            },
            basic_edge/.style={
                ->,
                draw=black,
                very thick,
            },
            edge_weight/.style={
                font=\footnotesize
            }]
                \node (a) [basic,source_target] at (0,2.5) {s};
                \node (b) [basic] at (1,0.75) {};
                \node (c) [basic] at (1.5,4) {};
                \node (d) [basic] at (2,2.25) {};
                \node (e) [basic] at (2.75,0.25) {};
                \node (f) [basic] at (3.25,4.5) {};
                \node (g) [basic] at (3.75,3) {};
                \node (h) [basic] at (4.25,1.5) {};
                \node (i) [basic] at (5,4) {};
                \node (j) [basic] at (5.75,2.25) {};
                \node (k) [basic] at (6,0.5) {};
                \node (l) [basic] at (7.25,3.5) {};
                \node (m) [basic] at (7.5,1.5) {};
                \node (n) [basic,source_target] at (8.75,2.5) {t};
                
                \draw (a) edge [basic_edge] node[edge_weight,anchor=west]{3} (b);
                \draw (a) edge [basic_edge] node[edge_weight,anchor=south]{3} (c);
                \draw (a) edge [basic_edge] node[edge_weight,anchor=south]{2} (d);
                \draw (b) edge [basic_edge] node[edge_weight,anchor=south]{2} (d);
                \draw (b) edge [basic_edge] node[edge_weight,anchor=south]{3} (e);
                \draw (c) edge [basic_edge] node[edge_weight,anchor=east]{2} (d);
                \draw (c) edge [basic_edge] node[edge_weight,anchor=south]{2} (f);
                \draw (c) edge [basic_edge] node[edge_weight,anchor=south]{4} (g);
                \draw (d) edge [basic_edge] node[edge_weight,anchor=south]{4} (h);
                \draw (d) edge [basic_edge] node[edge_weight,anchor=south]{3} (g);
                \draw (d) edge [basic_edge] node[edge_weight,anchor=west]{3} (e);
                \draw (e) edge [basic_edge] node[edge_weight,anchor=south]{2} (h);
                \draw (e) edge [basic_edge] node[edge_weight,anchor=south]{2} (k);
                \draw (f) edge [basic_edge] node[edge_weight,anchor=west]{2} (g);
                \draw (f) edge [basic_edge] node[edge_weight,anchor=south]{2} (i);
                \draw (g) edge [basic_edge] node[edge_weight,anchor=south]{3} (i);
                \draw (g) edge [basic_edge] node[edge_weight,anchor=west]{2} (h);
                \draw (g) edge [basic_edge] node[edge_weight,anchor=south]{2} (j);
                \draw (h) edge [basic_edge] node[edge_weight,anchor=south]{1} (j);
                \draw (h) edge [basic_edge] node[edge_weight,anchor=south]{3} (k);
                \draw (i) edge [basic_edge] node[edge_weight,anchor=west]{3} (j);
                \draw (i) edge [basic_edge] node[edge_weight,anchor=south]{3} (l);
                \draw (j) edge [basic_edge] node[edge_weight,anchor=west]{2} (k);
                \draw (j) edge [basic_edge] node[edge_weight,anchor=south]{3} (l);
                \draw (j) edge [basic_edge] node[edge_weight,anchor=south]{3} (m);
                \draw (j) edge [basic_edge] node[edge_weight,anchor=south]{4} (n);
                \draw (k) edge [basic_edge] node[edge_weight,anchor=south]{2} (m);
                \draw (l) edge [basic_edge] node[edge_weight,anchor=south]{3} (n);
                \draw (m) edge [basic_edge] node[edge_weight,anchor=north]{3} (n);
        \end{tikzpicture}
        } 
        \caption{An example weighted, directed graph, $G = (V,E)$, with non-negative edge weights, source node, $s$, and target node, $t$.}
        \label{fig:directed_weighted_graph}
\end{figure}
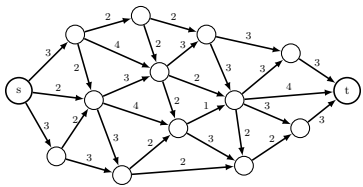

Figure \ref{fig:directed_weighted_graph} illustrates an example weighted, directed graph with non-negative edge weights, a source node labeled $s$, and target node labeled $t$.
Let us denote a shortest path from $s$ to $t$ as $p_{s,t}$, with total length or cost $c(p_{s,t})$, and let us introduce the notation $\tilde{p}_{s,t}$ to refer to a specific choice of shortest path, in case there is a tie and the shortest path $p_{s,t}$ is not unique.
It is widely known that $p_{s,t}$ has a property known as \emph{optimal substructure} \cite{cormen2001introduction}.
\begin{thm}[Optimal substructure]\label{thm:opt_sub}
If node $m \in p_{s,v}$, then $p_{s,v} = p_{s,m} \| p_{m,v}$, where here $\|$ denotes the concatenation of the paths' sequences of edges.
\end{thm}

Let $b_s(v)$ denote a particular intermediate node---the \emph{predecessor} of node $v$ on the path $\tilde{p}_{s,v}$. The predecessor is the node $b_s(v)$ such that,
\begin{equation}\label{eq:recursive_predecessor_path}
\tilde{p}_{s,v} = \tilde{p}_{s,b_s(v)} \| (b_s(v),v).
\end{equation}
\noindent From the fact that $p_{s,v}$ is a shortest path, it follows that,
\begin{equation}\label{eq:forward_cost}
c(p_{s,v}) = \min_{u \in N^-(v)} [c(p_{s,u}) + c(u,v)],
\end{equation}
\begin{equation}\label{eq:predecessor}
b_s(v) = \operatorname*{arg\,min}_{u \in N^-(v)} [c(p_{s,u}) + c(u,v)],
\end{equation}
\noindent where $N^-(v)$ is the in-neighborhood of $v$.

The set of all predecessors $\{b_s(v)|v \in V\}$ can be represented graphically as a \emph{shortest path tree} (SPT). We formally define a \emph{predecessor} SPT as a directed graph $G_\beta = (V,E_\beta)$, where,

\begin{equation}\label{eq:predecessor_graph}
E_\beta = \{(v,b_s(v)) | v \in V\}.
\end{equation}

\begin{figure}
    \centering
        \resizebox{0.4\textwidth}{!}{
        \begin{tikzpicture} [
            basic/.style={
                circle,
                minimum size=5mm,
                draw=black,
                font=\tiny,
                thick
            },
            source_target/.style={
                minimum size=7mm,
                very thick,
                font=\small,
            },
            basic_edge/.style={
                ->,
                draw=black,
                very thick,
            },
            edge_weight/.style={
                font=\footnotesize
            }]
                \node (a) [basic,source_target] at (0,2.5) {s};
                \node (b) [basic] at (1,0.75) {3};
                \node (c) [basic] at (1.5,4) {3};
                \node (d) [basic] at (2,2.25) {2};
                \node (e) [basic] at (2.75,0.25) {5};
                \node (f) [basic] at (3.25,4.5) {5};
                \node (g) [basic] at (3.75,3) {5};
                \node (h) [basic] at (4.25,1.5) {6};
                \node (i) [basic] at (5,4) {7};
                \node (j) [basic] at (5.75,2.25) {7};
                \node (k) [basic] at (6,0.5) {7};
                \node (l) [basic] at (7.25,3.5) {10};
                \node (m) [basic] at (7.5,1.5) {9};
                \node (n) [basic] at (8.75,2.5) {11};
                \draw (b) edge [basic_edge] node[edge_weight,anchor=west]{3} (a);
                \draw (c) edge [basic_edge] node[edge_weight,anchor=south]{3} (a);
                \draw (d) edge [basic_edge] node[edge_weight,anchor=south]{2} (a);
                \draw (f) edge [basic_edge] node[edge_weight,anchor=south]{2} (c);
                \draw (h) edge [basic_edge] node[edge_weight,anchor=south]{4} (d);
                \draw (g) edge [basic_edge] node[edge_weight,anchor=south]{3} (d);
                \draw (e) edge [basic_edge] node[edge_weight,anchor=west]{3} (d);
                \draw (i) edge [basic_edge] node[edge_weight,anchor=south]{2} (f);
                \draw (j) edge [basic_edge] node[edge_weight,anchor=south]{2} (g);
                \draw (k) edge [basic_edge] node[edge_weight,anchor=south]{2} (e);
                \draw (l) edge [basic_edge] node[edge_weight,anchor=south]{3} (i);
                \draw (m) edge [basic_edge] node[edge_weight,anchor=south]{2} (k);
                \draw (n) edge [basic_edge] node[edge_weight,anchor=south]{4} (j);
        \end{tikzpicture}
        } 
        \caption{A predecessor shortest path tree, $G_\beta = (V,E_\beta)$, as described in Section \ref{sec:shortestpathtree}. The source node $s$ is labeled, as well as the edge weights and the cost $c(\tilde{p}_{s,v})$ of the shortest path $\tilde{p}_{s,v}$ for each node $v \in V$.}
        \label{fig:shortest_path_tree}
\end{figure}
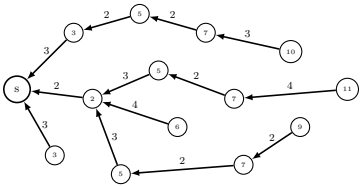

\noindent Each edge $(v,b_s(v)) \in E_\beta$ is known as a \emph{back pointer}.
Within the context of an SPT, each back pointer is a pointer to a node's parent.
In the case of a tie between multiple paths for $p_{s,v}$, an SPT defines a particular choice of best path, and throughout this paper we denote the particular choice defined by an SPT as $\tilde{p}_{s,v}$.
Dijkstra's algorithm \cite{dijkstra1959note} implemented with a Fibonacci heap \cite{fredman1987fibonacci} has the best-known worst-case complexity, of $O(|E| + |V| \log |V|)$, on arbitrary graphs for computing a SPT.

\subsubsection{Via-path}\label{sec:viapath}

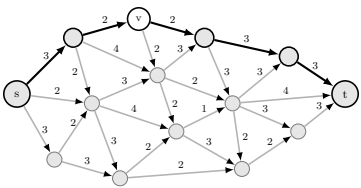
\begin{figure}
    \centering
        \resizebox{0.4\textwidth}{!}{
        \begin{tikzpicture} [
            active/.style={
                circle,
                minimum size=5mm,
                draw=black,
                fill=black!10,
                font=\tiny,
                very thick
            },
            inactive/.style={
                circle,
                minimum size=4mm,
                draw=black!50,
                fill=black!10,
                font=\tiny,
                thick
            },
            source_target/.style={
                minimum size=7mm,
                very thick,
                font=\small,
            },
            via/.style={
                minimum size=6mm,
                very thick,
                fill=white,
                font=\footnotesize
            },
            active_edge/.style={
                ->,
                draw=black,
                ultra thick,
            },
            inactive_edge/.style={
                ->,
                draw=black!30,
                very thick,
            },
            edge_weight/.style={
                font=\footnotesize
            }]
                \node (a) [active,source_target] at (0,2.5) {s};
                \node (b) [inactive] at (1,0.75) {};
                \node (c) [active] at (1.5,4) {};
                \node (d) [inactive] at (2,2.25) {};
                \node (e) [inactive] at (2.75,0.25) {};
                \node (f) [active,via] at (3.25,4.5) {v};
                \node (g) [inactive] at (3.75,3) {};
                \node (h) [inactive] at (4.25,1.5) {};
                \node (i) [active] at (5,4) {};
                \node (j) [inactive] at (5.75,2.25) {};
                \node (k) [inactive] at (6,0.5) {};
                \node (l) [active] at (7.25,3.5) {};
                \node (m) [inactive] at (7.5,1.5) {};
                \node (n) [active,source_target] at (8.75,2.5) {t};
                
                \draw (a) edge [inactive_edge] node[edge_weight,anchor=west]{3} (b);
                \draw (a) edge [active_edge] node[edge_weight,anchor=south]{3} (c);
                \draw (a) edge [inactive_edge] node[edge_weight,anchor=south]{2} (d);
                \draw (b) edge [inactive_edge] node[edge_weight,anchor=south]{2} (d);
                \draw (b) edge [inactive_edge] node[edge_weight,anchor=south]{3} (e);
                \draw (c) edge [inactive_edge] node[edge_weight,anchor=east]{2} (d);
                \draw (c) edge [active_edge] node[edge_weight,anchor=south]{2} (f);
                \draw (c) edge [inactive_edge] node[edge_weight,anchor=south]{4} (g);
                \draw (d) edge [inactive_edge] node[edge_weight,anchor=south]{4} (h);
                \draw (d) edge [inactive_edge] node[edge_weight,anchor=south]{3} (g);
                \draw (d) edge [inactive_edge] node[edge_weight,anchor=west]{3} (e);
                \draw (e) edge [inactive_edge] node[edge_weight,anchor=south]{2} (h);
                \draw (e) edge [inactive_edge] node[edge_weight,anchor=south]{2} (k);
                \draw (f) edge [inactive_edge] node[edge_weight,anchor=west]{2} (g);
                \draw (f) edge [active_edge] node[edge_weight,anchor=south]{2} (i);
                \draw (g) edge [inactive_edge] node[edge_weight,anchor=south]{3} (i);
                \draw (g) edge [inactive_edge] node[edge_weight,anchor=west]{2} (h);
                \draw (g) edge [inactive_edge] node[edge_weight,anchor=south]{2} (j);
                \draw (h) edge [inactive_edge] node[edge_weight,anchor=south]{1} (j);
                \draw (h) edge [inactive_edge] node[edge_weight,anchor=south]{3} (k);
                \draw (i) edge [inactive_edge] node[edge_weight,anchor=west]{3} (j);
                \draw (i) edge [active_edge] node[edge_weight,anchor=south]{3} (l);
                \draw (j) edge [inactive_edge] node[edge_weight,anchor=west]{2} (k);
                \draw (j) edge [inactive_edge] node[edge_weight,anchor=south]{3} (l);
                \draw (j) edge [inactive_edge] node[edge_weight,anchor=south]{3} (m);
                \draw (j) edge [inactive_edge] node[edge_weight,anchor=south]{4} (n);
                \draw (k) edge [inactive_edge] node[edge_weight,anchor=south]{2} (m);
                \draw (l) edge [active_edge] node[edge_weight,anchor=south]{3} (n);
                \draw (m) edge [inactive_edge] node[edge_weight,anchor=south]{3} (n);
        \end{tikzpicture}
        } 
        \caption{An example via-path, $p_{s,v,t}$, which is the shortest path from $s$ to $t$ that passes through node $v$, and can be decomposed into $p_{s,v} \| p_{v,t}$.}
        \label{fig:via_path}
    
\end{figure}

\begin{dfn}[Via-path]
A \emph{via-path}, denoted $p_{s,v,t}$, is the shortest path in a graph, $G = (V,E)$, from a source node, $s$, to a target node, $t$, that passes through an intermediate node, $v$.
\end{dfn}

By Theorem \ref{thm:opt_sub}, $p_{s,v,t} = p_{s,v} \| p_{v,t}$, and similarly, $c(p_{s,v,t}) = c(p_{s,v}) + c(p_{v,t})$. 
A path of this type was first described as a \emph{gateway path} \cite{lombard1993gateway,scaparra2014corridor}, with $v$ known as the path's \emph{gateway node}.
Others \cite{fujita2003dual,jones2012method} refer to node $v$ as a \emph{via point}.
The quantity $c(p_{s,v,t})$ has sometimes been referred to as the conditional minimum transit cost (CMTC) \cite{pinto2009beyond}.
A recent number of works in alternative route suggestion for road networks \cite{abraham2013alternative,luxen2012candidate} refer to a path of this type as a \emph{via-path} and node $v$ as a \emph{via-node}. In order to help standardize the terminology, we choose to adopt and expand upon this recent terminology.

\subsubsection{Cascading via-paths}\label{sec:allviapaths}

To compute the via-path from $s$ to $t$ for \emph{every} node $v \in V$, two \emph{sets} of paths need to be computed---$\{p_{s,v}|v \in V\}$ and $\{p_{v,t}|v \in V\}$.
Much like the predecessor SPT rooted at $s$ compactly encodes the set of paths $\{\tilde{p}_{s,v}|v \in V\}$, the set of paths $\{\tilde{p}_{v,t}|v \in V\}$ are compactly encoded within a \emph{successor} SPT, rooted at target node $t$.
Let $f_t(v) \in p_{s,v}$ denote the \emph{successor} of node $v$ on the path $p_{v,t}$.
In other words, the successor of $v$ is the node $f_t(v)$ such that,

\begin{equation}\label{eq:recursive_successor_path}
p_{v,t} = (v,f_t(v)) \| p_{f_t(v),t}.
\end{equation}

\noindent From the fact that $p_{v,t}$ is a shortest path, it follows that,

\begin{equation}\label{eq:backward_cost}
c(p_{v,t}) = \min_{w \in N^+(v)} [c(v,w) + c(p_{w,t})],
\end{equation}

\begin{equation}\label{eq:successor}
f_t(v) = \operatorname*{arg\,min}_{w \in N^+(v)} [c(v,w) + c(p_{w,t})],
\end{equation}

\noindent where $N^+(v)$ is the out-neighborhood of $v$.

\begin{figure}
    \centering
        \resizebox{0.4\textwidth}{!}{
        \begin{tikzpicture} [
            basic/.style={
                circle,
                minimum size=5mm,
                draw=black,
                font=\tiny,
                thick
            },
            source_target/.style={
                minimum size=7mm,
                very thick,
                font=\small,
            },
            basic_edge/.style={
                ->,
                draw=black,
                very thick,
            },
            edge_weight/.style={
                font=\footnotesize
            }]
                \node (a) [basic] at (0,2.5) {11};
                \node (b) [basic] at (1,0.75) {10};
                \node (c) [basic] at (1.5,4) {10};
                \node (d) [basic] at (2,2.25) {9};
                \node (e) [basic] at (2.75,0.25) {7};
                \node (f) [basic] at (3.25,4.5) {8};
                \node (g) [basic] at (3.75,3) {6};
                \node (h) [basic] at (4.25,1.5) {5};
                \node (i) [basic] at (5,4) {6};
                \node (j) [basic] at (5.75,2.25) {4};
                \node (k) [basic] at (6,0.5) {5};
                \node (l) [basic] at (7.25,3.5) {3};
                \node (m) [basic] at (7.5,1.5) {3};
                \node (n) [basic,source_target] at (8.75,2.5) {t};

                \draw (a) edge [basic_edge] node[edge_weight,anchor=south]{2} (d);
                \draw (b) edge [basic_edge] node[edge_weight,anchor=south]{3} (e);
                \draw (c) edge [basic_edge] node[edge_weight,anchor=south]{2} (f);
                \draw (d) edge [basic_edge] node[edge_weight,anchor=south]{3} (g);
                \draw (e) edge [basic_edge] node[edge_weight,anchor=south]{2} (k);
                \draw (f) edge [basic_edge] node[edge_weight,anchor=south]{2} (i);
                \draw (g) edge [basic_edge] node[edge_weight,anchor=south]{2} (j);
                \draw (h) edge [basic_edge] node[edge_weight,anchor=south]{1} (j);
                \draw (i) edge [basic_edge] node[edge_weight,anchor=south]{3} (l);
                \draw (j) edge [basic_edge] node[edge_weight,anchor=south]{4} (n);
                \draw (k) edge [basic_edge] node[edge_weight,anchor=south]{2} (m);
                \draw (l) edge [basic_edge] node[edge_weight,anchor=south]{3} (n);
                \draw (m) edge [basic_edge] node[edge_weight,anchor=north]{3} (n);
        \end{tikzpicture}
        } 
        \caption{A successor shortest path tree, $G_\phi = (V,E_\phi)$, as described in Section \ref{sec:allviapaths}. The target node $t$ is labeled, as well as the edge weights and the cost $c(\tilde{p}_{v,t})$ of the shortest path $\tilde{p}_{v,t}$ for each node $v \in V$.}
        \label{fig:successor_shortest_path_tree}
\end{figure}
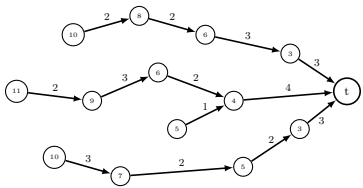

The successor shortest path tree represents the set of all successors $\{f_t(v)|v \in V\}$ graphically.
We formally define a successor shortest path tree as a directed graph $G_\phi = (V,E_\phi)$, where,

\begin{equation}
E_\phi = \{(v,f_t(v)) | v \in V\}.
\end{equation}

\noindent Each edge $(v,f_t(v)) \in E_\phi$ is known as a \emph{forward pointer} and just as in the predecessor shortest path tree, each forward pointer in the successor shortest path tree is considered a pointer to a node's parent.

We refer to the set of via-paths defined by $G_\beta$ and $G_\phi$ as the set of \emph{cascading via-paths} (CVPs).
This name is due to the chained or cascading nature of the optimal paths represented by the two trees.
As a cascading via-path is a particular choice of via-path, $p_{s,v,t}$, for a node $v$, which is specifically defined by the two shortest path trees and may be just one of several paths tied for best via-path, we denote it with special notation as $\tilde{p}_{s,v,t}$.
Even though tie-breaking may be done arbitrarily at each node when computing $G_\beta$ and $G_\phi$, the set of CVPs has more structure than a random choice of best via-path for each node, as the trees introduce dependency between the CVP of a node and those of its predecessor and successor.

\begin{figure}
    \centering
        \resizebox{0.4\textwidth}{!}{
        \begin{tikzpicture} [
            basic/.style={
                circle,
                minimum size=5mm,
                draw=black,
                font=\tiny,
                thick
            },
            source_target/.style={
                minimum size=7mm,
                very thick,
                font=\small,
            },
            basic_edge/.style={
                ->,
                draw=black,
                very thick,
            },
            reciprocal_edge/.style={
                ->,
                bend left=10,
                draw=black,
                very thick,
            },
            back_pointer/.style={
                densely dashed
            },
            edge_weight/.style={
                font=\footnotesize
            }]
                \node (a) [basic,source_target] at (0,2.5) {s};
                \node (b) [basic,fill=black!8] at (1,0.75) {12};
                \node (c) [basic,fill=black!16] at (1.5,4) {13};
                \node (d) [basic] at (2,2.25) {11};
                \node (e) [basic,fill=black!8] at (2.75,0.25) {12};
                \node (f) [basic,fill=black!16] at (3.25,4.5) {13};
                \node (g) [basic] at (3.75,3) {11};
                \node (h) [basic] at (4.25,1.5) {11};
                \node (i) [basic,fill=black!16] at (5,4) {13};
                \node (j) [basic] at (5.75,2.25) {11};
                \node (k) [basic,fill=black!8] at (6,0.5) {12};
                \node (l) [basic,fill=black!16] at (7.25,3.5) {13};
                \node (m) [basic,fill=black!8] at (7.5,1.5) {12};
                \node (n) [basic,source_target] at (8.75,2.5) {t};
                
                \draw (b) edge [basic_edge,back_pointer] node[edge_weight,anchor=west]{3} (a);
                \draw (c) edge [basic_edge,back_pointer] node[edge_weight,anchor=south]{3} (a);
                \draw (d) edge [reciprocal_edge,back_pointer] node[edge_weight,anchor=north]{2} (a);
                \draw (f) edge [reciprocal_edge,back_pointer] node[edge_weight,anchor=north]{2} (c);
                \draw (h) edge [basic_edge,back_pointer] node[edge_weight,anchor=north]{4} (d);
                \draw (g) edge [reciprocal_edge,back_pointer] node[edge_weight,anchor=north]{3} (d);
                \draw (e) edge [basic_edge,back_pointer] node[edge_weight,anchor=west]{3} (d);
                \draw (i) edge [reciprocal_edge,back_pointer] node[edge_weight,anchor=north]{2} (f);
                \draw (j) edge [reciprocal_edge,back_pointer] node[edge_weight,anchor=north]{2} (g);
                \draw (k) edge [reciprocal_edge,back_pointer] node[edge_weight,anchor=north]{2} (e);
                \draw (l) edge [reciprocal_edge,back_pointer] node[edge_weight,anchor=north]{3} (i);
                \draw (m) edge [reciprocal_edge,back_pointer] node[edge_weight,anchor=north]{2} (k);
                \draw (n) edge [reciprocal_edge,back_pointer] node[edge_weight,anchor=north]{4} (j);
                
                \draw (a) edge [reciprocal_edge] node[edge_weight,anchor=south]{2} (d);
                \draw (b) edge [basic_edge] node[edge_weight,anchor=south]{3} (e);
                \draw (c) edge [reciprocal_edge] node[edge_weight,anchor=south]{2} (f);
                \draw (d) edge [reciprocal_edge] node[edge_weight,anchor=south]{3} (g);
                \draw (e) edge [reciprocal_edge] node[edge_weight,anchor=south]{2} (k);
                \draw (f) edge [reciprocal_edge] node[edge_weight,anchor=south]{2} (i);                
                \draw (g) edge [reciprocal_edge] node[edge_weight,anchor=south]{2} (j);
                \draw (h) edge [basic_edge] node[edge_weight,anchor=north]{1} (j);                
                \draw (i) edge [reciprocal_edge] node[edge_weight,anchor=south]{3} (l);
                \draw (j) edge [reciprocal_edge] node[edge_weight,anchor=south]{4} (n);
                \draw (k) edge [reciprocal_edge] node[edge_weight,anchor=south]{2} (m);
                \draw (l) edge [basic_edge] node[edge_weight,anchor=south]{3} (n);
                \draw (m) edge [basic_edge] node[edge_weight,anchor=north]{3} (n);
        \end{tikzpicture}
        } 
        \caption{This figure depicts $G_\beta \cup G_\phi$ for a source node $s$ and target node $t$, with edge set $E_\beta$ depicted here using dashed lines and edge set $E_\phi$ depicted using solid lines. Each node $v \in V$ has two outgoing edges, one dashed and one solid. Following the chain of dashed edges from $v$ to $s$ and the chain of solid edges from $v$ to $t$ reveals the CVP, $\tilde{p}_{s,v,t}$ for each node $v$. Each node in the figure contains a label indicating $c(\tilde{p}_{s,v,t})$, the cost of its CVP.}
        \label{fig:all_via_paths}
\end{figure}
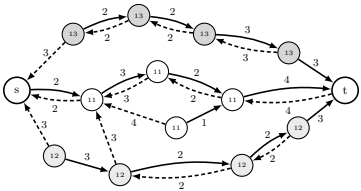

\subsection{Cascading via-path enumeration and graph partitioning}\label{sec:stage2}

The next stage (stage two) of the procedure that we describe is an efficient method for \emph{enumerating} the distinct CVPs given the two SPTs---in other words, enumerating $\Lambda$ given $G_\beta$ and $G_\phi$.
Any procedure for \emph{explicitly} representing the elements of $\Lambda$ must be at least $O(|V|^2)$, based on there being up to $|V|$ distinct CVPs, each of which is $O(|V|)$ in length.
However, the approach described here, called the \emph{reciprocal pointer chain method} \cite{smock2012reciprocal,smock2012efficient}, \emph{enumerates} $\Lambda$ without duplicates in $O(|V|)$ time---independent of the lengths of the paths.
This enumeration represents $\Lambda$ implicitly but in a way that lets us select the top $k$ elements without explicitly computing them, netting us a significant computational savings for computing the final set of $k$-best CVPs, $\hat{\Lambda}^{(k)}$, in stage four of the overall procedure, when $k \ll |\Lambda|$.

A related method for enumerating CVPs is the \emph{disjoint plateau method} \cite{camvit2006choice,jones2012method}, known simply as the \emph{plateau method} in prior work.
While in the original papers and subsequent road network literature this algorithm and the reciprocal pointer chains algorithm are considered equivalent, we show in this paper that differences between the two algorithms can arise when there are ties for best path $p_{s,v}$ or $p_{v,t}$ for a node $v \in V$.
In this case, the disjoint plateau method is only guaranteed to enumerate a subset of the full set, $\Lambda$.
We present the two algorithms in this section and sort out their differences from a theoretical perspective in more detail in Section \ref{sec:theory}.

\subsubsection{Reciprocal pointer chain method}\label{sec:rpc_method}

\begin{algorithm}[t]
	\caption{Partition $V$ into RPCs to enumerate the CVPs}
	\label{alg:partition}
	\centering
	\begin{algorithmic}[1]
	\footnotesize
		\Require $G_\beta = (V,E_\beta)$, the predecessor tree; $G_\phi = (V,E_\phi)$, the successor tree; $s$ is the source; $t$ is the target
		\Procedure{RPC}{$G_\beta$,$G_\phi$,$s$,$t$}
		\For {each node $v \in V$}
			\State $v$.visited $\gets$ FALSE \EndFor
		\While {there exists an unvisited node} \label{line:outer_loop}
			\State Choose an unvisited node $v$
			\State $v$.visited $\gets$ TRUE
			\State $p \gets$ [] \Comment{Start a new chain/partition}
			\State $p$.\Call{addToHead}{$v$} \Comment{$v$ is the first element of the chain/partition}
			\State $x \gets v$
			\While {$x$ $\ne$ $s$ \textbf{and} $G_\phi(G_\beta(x)$.parent).parent $= x$} \label{line:inner_loop1}
				\State $x \gets$ $G_\beta(x)$.parent
				\State $x$.visited $\gets$ TRUE
				\State $p$.\Call{addToHead}{$x$}
			\EndWhile
			\State $x \gets v$
				\While {$x$ $\ne$ $t$ \textbf{and} $G_\beta(G_\phi(x)$.parent).parent $= x$} \label{line:inner_loop2}
				\State $x \gets$ $G_\phi(x)$.parent
				\State $x$.visited $\gets$ TRUE
				\State $p$.\Call{addToTail}{$x$}
			\EndWhile
			\State $\Pi$.\Call{add}{$p$}
		\EndWhile
		\State \Return $\Pi$ \Comment {$\Pi$ represents the set of partitions}
		\EndProcedure
	\end{algorithmic}
\end{algorithm}

Pseudocode for the RPC method is given in Algorithm \ref{alg:partition}.
At a high level, the RPC method is a partitioning procedure that enumerates the number of distinct CVPs by instead grouping together nodes that have the same CVP, and enumerating the groups.
Each partition is referred to as a \emph{reciprocal pointer chain} (RPC), and we denote the set of RPCs as $\Pi$, where $|\Pi| = |\Lambda|$.
Grouping the nodes together can be done efficiently compared to explicitly enumerating $\Lambda$ because, as we prove in Section \ref{sec:theory}, Theorem \ref{thm:rp}, there is a simple, $O(1)$ test, called the reciprocal pointer chain test, to decide if for two nodes, $u$ and $v$, $\tilde{p}_{s,v,t} = \tilde{p}_{s,u,t}$.
Furthermore, by testing in a specific order---using a depth-first search---each node need only be tested a constant number of times to place it in a group.
The outer loop of the procedure, starting at Line \ref{line:outer_loop}, enumerates the partitions/chains one-at-a-time, creating and completing a single partition within each iteration.
It picks an unvisited node, $v$, at random and starts a new partition/chain at the beginning of the loop with $v$ as its first member.
Once the first node belonging to a partition is visited by the search, all of the remaining nodes belonging to the partition are immediately found by two new searches corresponding to the two inner loops of the procedure.

The first inner loop, starting at Line \ref{line:inner_loop1}, iteratively adds nodes to the head of the current chain.
It adds the predecessor of the head of the chain to the head of the chain if the current head is not the source node $s$ and if the predecessor of the current head has the current head as its successor.
This second condition is what we call a \emph{reciprocal pointer} between the two nodes.
Thus all pairs of neighboring nodes in the chain have reciprocal pointers between them.
The second inner loop, starting at Line \ref{line:inner_loop2}, is identical to the first inner loop but in the opposite direction from $v$.
It iteratively adds the successor of the tail of the chain to the tail if the current tail is not the target node $t$ and if the successor of the current tail has the current tail as its predecessor.

An equivalent algorithm, which may make the behavior of the RPC algorithm clearer, is given by the following.
Let us define an undirected graph, $G_\gamma = (V,E_\gamma)$, such that,
\begin{equation}\label{eq:rpc_graph}
E_\gamma = \{(v,w) | (w,v) \in E_\beta \land (v,w) \in E_\phi\}.
\end{equation}
\noindent The partitioning of nodes produced by the RPC method corresponds to the connected components of $G_\gamma$.
Described in this way, it should be clearer that there can be RPCs that contain a single node and no edges.

\subsubsection{Disjoint plateau method}\label{sec:disjoint_plateau_method}

The disjoint plateau method \cite{jones2012method} differs from the RPC method given in Algorithm \ref{alg:partition} at just two lines---lines \ref{line:inner_loop1} and \ref{line:inner_loop2}.
At line \ref{line:inner_loop1}, while the RPC method checks to see if there is a reciprocal pointer between node $x$ and its predecessor, which we shall call $w$, the disjoint plateau method checks to make sure two things are true: that $w$ has not been visited and that $c(\tilde{p}_{s,x,t}) = c(\tilde{p}_{s,w,t})$.
The explicit check to make sure $w$ has not been visited, which is not needed by the RPC method, ensures that the partitions created by the disjoint plateau method are indeed disjoint.
Similarly, at line \ref{line:inner_loop2}, the disjoint plateau method checks to make sure the predecessor of $x$---which we denote $y$---has not been visited and that $c(\tilde{p}_{s,x,t}) = c(\tilde{p}_{s,y,t})$.

The main problem with this method for the purpose of enumerating CVPs is that the two conditions that 1) $w$ is the predecessor of $x$, and 2) $c(\tilde{p}_{s,x,t}) = c(\tilde{p}_{s,w,t})$, are together not enough to ensure that $x$ and $w$ have the same CVP, as encoded within the two shortest path trees.
For both this algorithm and the RPC algorithm, the number of CVPs enumerated corresponds to the number of executions of the outer loop.
But since the disjoint plateau method uses an asymmetric criterion for putting two nodes in the same partition, whether or not the two nodes are placed together depends on the order in which they are encountered during the depth-first search.
This potentially reduces the number of CVPs identified by the plateau method relative to the RPC method, and makes the solution computed by the disjoint plateau not unique given $G_\beta$ and $G_\phi$, since it is dependent on an arbitrary ordering of the nodes.
Thus between the two, only the RPC method is correct for enumerating $\Lambda$ for arbitrary graphs.

\subsection{Via-path ranking and sub-selection}\label{sec:stage3}

The third stage of the procedure for computing the $k$-best CVPs is to first sort the distinct paths, $\Pi$, by some measure of goodness and select the $k$ best paths, where $k$ is either explicitly specified or describes the number of paths satisfying a specified threshold.
The three measures we consider here all have in common that they can be computed in $O(1)$ (constant) time per path (independent of path length), given $\Pi$, $G_\beta$, and $G_\phi$.
Despite their computational simplicity, these measures result in a number of useful and interesting properties for the set of $k$ best CVPs, which we discuss in detail in Section \ref{sec:theory}, and demonstrate in practice in Section \ref{sec:results}.

\subsubsection{Path cost/length}

The first measure of goodness we consider for ordering $\hat{\Pi}$ is the simple cost, or length, associated with each path, which for via-node, $v$, is given by,
\begin{equation}
c(p_{s,v,t}) = c(p_{s,v}) + c(p_{v,t}).
\end{equation}
\noindent The values of $c(p_{s,v})$ and $c(p_{v,t})$ are computed during the computation of the two shortest path trees, which means $c(p_{s,v,t})$ can be computed in $O(1)$ additional time.
We refer to the ordered subset of paths created by ranking and sub-selecting with this measure as the $k$ shortest CVPs.


\subsubsection{Via-node fraction}\label{sec:vnf}

The second measure we consider for ranking CVPs is the \emph{via-node fraction}, $\omega(\tilde{p}_{s,v,t})$.
Let $|p|$ denote the number of edges, and therefore $|p|+1$ the number of nodes, in a path $p$.
Let $\tilde{p}_{h,l}$ represent the RPC from node $h$ to node $l$ associated with $\tilde{p}_{s,v,t}$.
The via-node fraction is given by,
\begin{equation}
\omega(\tilde{p}_{s,v,t}) = \frac{|\tilde{p}_{h,l}|+1}{|\tilde{p}_{s,v,t}|+1}.
\end{equation}
\noindent The via-node fraction can be interpreted either as the fraction of nodes within a CVP that are in its RPC, or as the fraction of nodes within $\tilde{p}_{s,v,t}$ whose CVP is $\tilde{p}_{s,v,t}$.
The quantities $|\tilde{p}_{h,l}|$ and $|\tilde{p}_{s,v,t}|$ have not been mentioned as being computed in the previous stages.
However, they can be computed analogously to $c(\tilde{p}_{h,l})$ and $c(\tilde{p}_{s,v,t})$, which means they can be computed from the values $|\tilde{p}_{s,h}|$, $|\tilde{p}_{s,v}|$,$|\tilde{p}_{v,t}|$, and $|\tilde{p}_{l,t}|$.
These values can be computed and cached during shortest path tree computation, with $O(1)$ additional time and space per node.
For each node $v \in V$, we compute,
\begin{equation}
|\tilde{p}_{s,v}| = |\tilde{p}_{s,b_s(v)}| + 1,
\end{equation}
\begin{equation}
|\tilde{p}_{v,t}| = |\tilde{p}_{f_t(v),t}| + 1,
\end{equation}
\noindent where $|\tilde{p}_{s,s}| = |\tilde{p}_{t,t}| = 0$.


\subsubsection{RPC cost fraction}\label{sec:rpccf}

A third measure, very similar to via-node fraction, but using the edges instead of the nodes and weighting the edges by their cost, is the RPC cost fraction.
We define the RPC cost fraction, $\rho(\tilde{p}_{s,v,t})$, for path $\tilde{p}_{s,v,t}$, as,
\begin{equation}
\rho(\tilde{p}_{s,v,t}) = \frac{c(\tilde{p}_{h,l})}{c(\tilde{p}_{s,v,t})} = \frac{c(\tilde{p}_{s,l}) - c(\tilde{p}_{s,h})}{c(\tilde{p}_{s,v}) + c(\tilde{p}_{v,t})} = \frac{c(\tilde{p}_{h,t}) - c(\tilde{p}_{l,t})}{c(\tilde{p}_{s,v}) + c(\tilde{p}_{v,t})}.
\end{equation}
\noindent Using values cached during the computation of the two shortest path trees, $\rho(\tilde{p}_{s,v,t})$ can be computed in $O(1)$ additional time per path.
Ranking by RPC cost fraction has essentially the same theoretical justification as ranking by via-node fraction, though, the two can have significantly different outcomes in practice.

\subsection{Via-path extraction}\label{sec:stage4}

The final stage of the procedure we have outlined is, for each of the implicit $k$-best CVPs identified in earlier stages, to extract its explicit sequence of nodes from its implicit representation within the two shortest path trees. This produces the ordered set of $k$-best CVPs, $\hat{\Lambda}^{(k)}$. This set is not to be confused with $\hat{\Pi}^{(k)}$, which represents the same paths using their RPCs but does not indicate the full sequence of nodes in each path. The RPCs represented in set $\Pi$ are node-disjoint, whereas the CVPs in set $\Lambda$ are not.

\begin{algorithm}[t]
	\caption{Algorithm for extracting each CVP (completing each RPC)}
	\label{alg:extraction}
	\begin{algorithmic}[1]
	\footnotesize
		\Require $G_\beta = (V,E_\beta)$, the predecessor tree; $G_\phi = (V,E_\phi)$, the successor tree; $s$ is the source; $t$ is the target; $\Pi$ is the set of RPCs	
		\Procedure{ExtractCVPs}{$G_\beta$,$G_\phi$,$s$,$t$,$\Pi$}
		\For {each $r \in \Pi$}
			\State $x$ $\gets$ $r$.\Call{getHead}{}
			\While {$x$ $\ne$ $s$}
				\State $x$ $\gets$ $G_\beta$($x$).parent
				\State $r$.\Call{addToHead}{$x$}
			\EndWhile
			\State $x$ $\gets$ $r$.\Call{getTail}{}
			\While {$x$ $\ne$ $t$}
				\State $x$ $\gets$ $G_\phi$($x$).parent
				\State $r$.\Call{addToTail}{$x$}
			\EndWhile
			\State $\Lambda$.\Call{add}{$r$}
		\EndFor
		\State \Return $\Lambda$ \Comment {$\Lambda$ represents the set of explicit cascading via-paths}
		\EndProcedure
	\end{algorithmic}
\end{algorithm}

Since each RPC corresponds to a distinct CVP, extraction corresponds to \emph{completing} each RPC---in other words, extending each RPC backward in the predecessor shortest path tree until it meets $s$ and forward in the successor shortest path tree until it meets $t$.
The procedure for doing this is given in Algorithm \ref{alg:extraction}.
The worst-case computational complexity of this stage of the overall CVP procedure is $O(k|V|)$, where $k \le |V|$.

\section{Theoretical Results and Implications}\label{sec:theory}

In this section we present the theory underlying cascading via-paths and reciprocal pointer chains from both a mathematical and a computational perspective. 
Our focus is on contributions that highlight properties we believe are of particular interest in application and therefore, throughout this section, we comment on and call attention to the particular importance of each of the theorems as we present them.
We then conclude the section with a summary of implications.

\subsection{Basic properties of via-paths and cascading via-paths}

\begin{prp}\label{thm:via_path_loop}
The shortest path from $s$ to $t$ passing through a node $v$ can contain a loop.
\end{prp}
\begin{proof}
An example of a via-path with a loop is given in Figure \ref{fig:via_path_loop_example}. In this example, for nodes $u$ and $v$, $u \in p_{s,v}$ and $u \in p_{v,t}$, which means that $p_{s,v,t} = p_{s,u} \| p_{u,v} \| p_{v,u} \| p_{u,t}$. Segment $p_{u,v} \| p_{v,u}$ is a loop.
\end{proof}

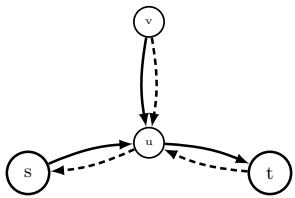
\begin{figure}
    \centering
        \resizebox{0.33\textwidth}{!}{
        \begin{tikzpicture} [
            basic/.style={
                circle,
                minimum size=5mm,
                draw=black,
                font=\tiny,
                thick
            },
            source_target/.style={
                minimum size=7mm,
                very thick,
                font=\small,
            },
            basic_edge/.style={
                ->,
                draw=black,
                very thick,
            },
            reciprocal_edge/.style={
                ->,
                bend left=10,
                draw=black,
                very thick,
            },
            back_pointer/.style={
                densely dashed
            },
            edge_weight/.style={
                font=\footnotesize
            }]
                \node (a) [basic,source_target] at (0,0) {s};
                \node (b) [basic] at (2,0.5) {u};
                \node (c) [basic] at (2,2.5) {v};
                \node (d) [basic,source_target] at (4,0) {t};
                
                \draw (b) edge [reciprocal_edge,back_pointer] node[edge_weight,anchor=north]{} (a);
                \draw (c) edge [reciprocal_edge,back_pointer] node[edge_weight,anchor=north]{} (b);
                \draw (d) edge [reciprocal_edge,back_pointer] node[edge_weight,anchor=north]{} (b);
            
                \draw (a) edge [reciprocal_edge] node[edge_weight,anchor=south]{} (b);
                \draw (b) edge [reciprocal_edge] node[edge_weight,anchor=south]{} (d);
                \draw (c) edge [reciprocal_edge,bend right=10] node[edge_weight,anchor=south]{} (b);

        \end{tikzpicture}
        } 
        \caption{This four-node graph example illustrates a case where $p_{s,v,t}$ contains a loop. This is discussed in Theorem \ref{thm:via_path_loop}.}
        \label{fig:via_path_loop_example}
\end{figure}

\noindent We comment that the shortest path between two nodes, $p_{s,t}$, does not contain a loop, but that the introduction of a via-node, $v$, that path $p_{s,v,t}$ must traverse introduces the possibility of a looped path. We note without proof that a via-path with a loop, such as $p_{s,v,t}$ in Figure \ref{fig:via_path_loop_example}, implies the existence of a via-path without a loop, such as $p_{s,u,t}$ in the same figure.

\begin{prp}\label{thm:cardinality}
$|\Lambda| \le |V|$.
\end{prp}
\begin{proof}
This follows directly from the fact that $\Lambda = \bigcup_{v \in V} \{\tilde{p}_{s,v,t}\}$.
\end{proof}

\noindent We include the preceding proposition to highlight that the size of the set of \emph{cascading} via-paths, $\Lambda$, as described in this paper, is tightly bounded, unlike the size of the superset of \emph{all} via-paths.

\begin{prp}
Path $\tilde{p}_{s,t} \in \Lambda$.
\end{prp}
\begin{proof}
Any path $q$, from $s$ to $t$, with $c(q) = c(\tilde{p}_{s,t})$, can be considered to be $\tilde{p}_{s,t}$. It is true that $\tilde{p}_{s,s,t} \in \Lambda$ and $\tilde{p}_{s,t,t} \in \Lambda$ from the definition of $\Lambda$. Because $c(\tilde{p}_{s,s}) = c(\tilde{p}_{t,t}) = 0$, we know that even if $\tilde{p}_{s,s,t} \ne \tilde{p}_{s,t,t}$, $c(\tilde{p}_{s,s,t}) = c(\tilde{p}_{s,t,t}) = c(\tilde{p}_{s,t})$. Therefore $\tilde{p}_{s,t} \in \Lambda$.
\end{proof}

\noindent Here we call attention to the fact that the set of CVPs, and thus the set of $k$ best CVPs given various optimality criteria, does include the shortest path from $s$ to $t$. This is important when comparing CVPs to other path sets such as the $k$ shortest disjoint paths, which in general does not guarantee inclusion of the single shortest path.

\subsection{Properties of reciprocal pointer chains}

In this section we describe an important feature of each cascading via-path---its \emph{reciprocal pointer chain}. As we show in succeeding sections, the RPC is useful not only for characterizing individual paths and properties of the set of CVPs as a whole, but also as a unique feature on its own.

\begin{dfn}[Reciprocal pointer]
A reciprocal pointer connects two nodes, $u$ and $v$, whenever either $(v,u) \in E_\beta$ and $(u,v) \in E_\phi$, or $(u,v) \in E_\beta$ and $(v,u) \in E_\phi$.
\end{dfn}

\begin{thm}[Reciprocal pointer theorem]\label{thm:rp}
If $u$ and $v$ are connected by a reciprocal pointer, then $\tilde{p}_{s,u,t} = \tilde{p}_{s,v,t}$.
\end{thm}
\begin{proof}
Assume without loss of generality from the definition of reciprocal pointer that $(v,u) \in E_\beta$ and $(u,v) \in E_\phi$. Then by the definition of a shortest path tree, $\tilde{p}_{s,v} = \tilde{p}_{s,u} \| (u,v)$ and $\tilde{p}_{u,t} = (u,v) \| \tilde{p}_{v,t}$. Then $\tilde{p}_{s,u,t} = \tilde{p}_{s,u} \| (u,v) \| \tilde{p}_{v,t} = \tilde{p}_{s,v,t}$.
\end{proof}

\begin{dfn}[Reciprocal pointer chain]
A reciprocal pointer chain is a sequence of nodes $(v_1,v_2,\dotsc,v_n)$ such that for all $1 \le i \le n-1: (v_{i+1},v_i) \in E_\beta$ and $(v_i,v_{i+1}) \in E_\phi$.
\end{dfn}

\noindent If there is a reciprocal pointer chain from $u$ to $w$ or from $w$ to $u$, then we say that $u$ and $w$ are \emph{connected} by a reciprocal pointer chain.

\begin{thm}[Reciprocal pointer chain theorem]\label{thm:rpc}
The statement $\tilde{p}_{s,u,t} = \tilde{p}_{s,v,t}$, for cascading via-paths $\tilde{p}_{s,u,t}$ and $\tilde{p}_{s,v,t}$, is equivalent to the statement that $u$ and $v$ are connected by a reciprocal pointer chain.
\end{thm}
\begin{proof}
By Theorem \ref{thm:rp} and induction, if $u$ and $v$ are connected by a reciprocal pointer chain, then $\tilde{p}_{s,u,t} = \tilde{p}_{s,v,t}$.
Next we prove the converse.
If $\tilde{p}_{s,u,t} = \tilde{p}_{s,v,t}$, then either $\tilde{p}_{s,u,t} = \tilde{p}_{s,u} \| \tilde{p}_{u,v} \| \tilde{p}_{v,t}$ or $\tilde{p}_{s,u,t} = \tilde{p}_{s,v} \| \tilde{p}_{v,u} \| \tilde{p}_{u,t}$.
Assume without loss of generality that $\tilde{p}_{s,u,t} = \tilde{p}_{s,u} \| \tilde{p}_{u,v} \| \tilde{p}_{v,t}$.
Then $u \in \tilde{p}_{s,v}$ and $v \in \tilde{p}_{u,t}$.
By the definition of a shortest path tree, $\forall (w,x) \in \tilde{p}_{s,v}: (x,w) \in E_\beta$.
Then because $u \in \tilde{p}_{s,v}$, $\forall (w,x) \in \tilde{p}_{u,v}: (x,w) \in E_\beta$.
By a similar argument, $\forall (w,x) \in \tilde{p}_{u,v}: (w,x) \in E_\phi$.
Thus $\forall (w,x) \in \tilde{p}_{u,v}: (x,w) \in E_\beta$ and $(w,x) \in E_\phi$, therefore the sequence of nodes in $\tilde{p}_{u,v}$ is a reciprocal pointer chain.
\end{proof}

\noindent Since a reciprocal pointer chain connects any two nodes that have the same cascading via-path, it must encompass a portion of the CVP itself.

\begin{cor}
The sequence of nodes in a reciprocal pointer chain from $u$ to $v$ is the same as in the shortest path $\tilde{p}_{u,v}$.
\end{cor}

\noindent Note, however, that the converse is not true in general.

\begin{cor}\label{thm:rpc_loopless}
A reciprocal pointer chain cannot contain a loop (repeated node). 
\end{cor}

\noindent This is notable because, though $p_{s,v,t}$ as a whole can contain a loop, Corollary \ref{thm:rpc_loopless} says there are potentially large subpaths of $p_{s,v,t}$ that cannot.

\begin{dfn}[Maximal reciprocal pointer chain]
A \emph{maximal} RPC is an RPC whose set of nodes is not a subset of the set of nodes of another RPC.
\end{dfn}

\noindent The set of RPCs, $\Pi$, computed by the RPC algorithm described in Section \ref{sec:rpc_method} is the set of maximal RPCs. We often refer to elements of this set as simply RPCs, though, rather than maximal RPCs, because non-maximal RPCs are not given much consideration.

\begin{thm}[Maximal RPC disjointness theorem]\label{thm:mrpc_disjointness}
For any two distinct maximal RPCs, $R_1$ and $R_2$, $|R_1 \cap R_2| = 0$.
\end{thm}
\begin{proof}
Suppose there is a node $u_0$ common to both RPCs. Then for an arbitrary node $u_1 \in R_1$, $\tilde{p}_{s,u_0,t} = \tilde{p}_{s,u_1,t}$. Similarly, for an arbitrary node $u_2 \in R_2$, it must be the case that $\tilde{p}_{s,u_0,t} = \tilde{p}_{s,u_2,t}$. As a result, $\tilde{p}_{s,u_1,t} = \tilde{p}_{s,u_2,t}$. Thus $u_1$ and $u_2$ must be part of the same RPC. But if every node in $R_1$ and $R_2$ are part of the same RPC, then $R_1 = R_2$, which contradicts the fact that they are distinct RPCs. Therefore, if two maximal RPCs are distinct, they must be completely disjoint.
\end{proof}

\begin{thm}[Maximal RPC uniqueness theorem]\label{thm:rpc_uniqueness}
Each node, $v \in V$, is a member of exactly one maximal RPC.
\end{thm}
\begin{proof}
A single node, $v$, by itself satisfies the definition of a reciprocal pointer chain---though it is a degenerate case without any reciprocal pointers, and a path length of 0.
Thus every node belongs to at least one RPC.
This RPC is either a maximal RPC, or there must be a larger RPC that is a superset of the non-maximal RPC.
By induction, there must be some largest RPC and because it is the largest, it cannot be the subset of any other RPC---thus it is a maximal RPC.
Therefore, every node belongs to at least one maximal RPC.
By Theorem \ref{thm:mrpc_disjointness}, $v$ can be a member of at most one maximal RPC.
Therefore, every node is a member of exactly one maximal RPC.
\end{proof}

\begin{thm}\label{thm:rpc_cvp}
The sets $\Lambda$ (the CVPs) and $\Pi$ (the maximal RPCs) are in a one-to-one relationship.
\end{thm}
\begin{proof}
By Theorem \ref{thm:rpc_uniqueness}, each node $v \in V$ is a member of exactly one maximal RPC. By Theorem \ref{thm:rpc}, all of the nodes in this maximal RPC have the same CVP, and all of the nodes in different maximal RPCs have a different CVP. Therefore, each distinct RPC corresponds to a distinct CVP.
\end{proof}

\noindent In other words, Theorem \ref{thm:rpc_cvp} indicates that by enumerating the distinct RPCs, we enumerate the distinct CVPs.
This proves the correctness of the RPC algorithm for enumerating CVPs described in Section \ref{sec:method}.

\begin{cor} \label{rpc_disjoint}
The set of maximal RPCs, $\Pi$, completely partitions the nodes of $V$ into $|\Pi| = |\Lambda|$ disjoint subsets.
\end{cor}

\indent This final corollary is primarily a re-wording of earlier theorems but we demonstrate its particular usefulness in Section \ref{sec:ksp_theory}, where we show it has implications for potentially greatly reducing the computation required to compute the $k$ shortest loopless paths.

\subsection{Properties of plateaus}\label{sec:plateau}

In this section, we discuss the \emph{plateau}, a term that has been used synonymously with \emph{reciprocal pointer chain} in prior work.
However, the concept of a plateau arises from a subtly different perspective, and we sort out the differences that can arise in general graphs.
Consider the following theorem:

\begin{thm}\label{thm:via_path_decrease_cost}
If node $u \in p_{s,v}$ or $u \in p_{v,t}$, for a pair of nodes $u$ and $v$, then $c(p_{s,u,t}) \le c(p_{s,v,t})$.
\end{thm}

\noindent The proof of Theorem \ref{thm:via_path_decrease_cost} is straightforward and is originally due to Shiller, et al.\cite{shiller2004computing}
This theorem forms the basis for the concept of a plateau, a term first coined by Jones \cite{jones2012method} but not rigorously defined.
A plateau is intended to encompass the following:

\begin{dfn}[Plateau]
A plateau is a sequence of nodes $(v_1,v_2,\dotsc,v_n)$ such that for all $1 \le i < n$, $(v_i,v_{i+1}) \in E$ and $c(p_{s,v_i,t}) = c(p_{s,v_{i+1},t})$.
\end{dfn}

\noindent This same concept is referred to as a \emph{valley} by Shiller, with both terms motivated by imagining a set of nodes residing in a two-dimensional planar graph, and via-path cost viewed as a three-dimensional surface above the plane.

Similar to the maximal RPC, we can define the concept of a maximal plateau.

\begin{dfn}[Maximal plateau]
A \emph{maximal} plateau is a plateau whose set of nodes is not a subset of the set of nodes of another plateau.
\end{dfn}

\noindent It follows from the definitions of plateau and reciprocal pointer chain that if a sequence of nodes $(v_1,v_2,\dotsc,v_n)$ corresponds to a maximal reciprocal pointer chain, it also corresponds to a maximal plateau---but the converse is not true.

\begin{prp}\label{thm:multi_plateau}
A node can be a member of more than one maximal plateau.
\end{prp}
\begin{proof}
We demonstrate this in Figure \ref{fig:all_via_paths}. There are two distinct paths from node $s$ to $t$ that have a total cost of 11, where for every node $v$ in the two paths, $c(p_{s,v,t}) = 11$. Both of these paths are maximal plateaus. Four nodes are common to both paths, therefore we have four examples of nodes that belong to more than one maximal plateau.
\end{proof}

\begin{cor}\label{thm:plateau_not_rpc}
A sequence of nodes $(v_1,v_2,\dotsc,v_n)$ that corresponds to a maximal plateau does not necessarily correspond to a maximal RPC.
\end{cor}
\begin{proof}
Since by Theorem \ref{thm:mrpc_disjointness}, any two maximal RPCs must be disjoint, then two non-disjoint maximal plateaus could not both be maximal RPCs.
\end{proof}

\noindent Like a via-path, a maximal plateau is an intrinsic feature of the graph and the choice of source and target nodes, but a maximal reciprocal pointer chain is extrinsic---it depends on $G_\beta$ and $G_\phi$, which are not necessarily unique (deterministic) given a graph $G$, source node $s$, and target node $t$.

Dissecting the disjoint plateau algorithm, it becomes clear that each plateau identified by the algorithm corresponds to the cascading via-path, $\tilde{p}_{s,v,t}$ for the node $v$ that was first added to the plateau.
Every node $u$ in the plateau is also in $\tilde{p}_{s,v,t}$.
But as a consequence of Proposition \ref{thm:multi_plateau} and Corollary \ref{thm:plateau_not_rpc}, a node $u$ can be in this plateau and it be the case that $\tilde{p}_{s,u,t} \ne \tilde{p}_{s,v,t}$.
In this case, a disjoint plateau may never be created that corresponds to CVP $\tilde{p}_{s,u,t}$.
Thus, unlike the set of reciprocal pointer chains, the set of disjoint plateaus is not guaranteed to enumerate all of the CVPs, $\Lambda = \{\tilde{p}_{s,v,t} | v \in V\}$, encoded in $G_\beta$ and $G_\phi$---enumerating only a subset in cases where it fails to enumerate the entire set.

\subsection{The $k$ shortest paths}\label{sec:ksp_theory}

By the definition of a via-path, $p_{s,v,t}$, for a node $v$, for any path $q$ such that $c(q) < c(p_{s,v,t})$, $v \notin q$.
From this definition it is possible to derive some interesting connections between via-paths and the $k$ shortest paths, including a way to efficiently reduce graph $G$ to a much smaller graph $G^\prime$ that yields the same $k$ shortest paths as $G$.

\begin{lem}\label{thm:kshortestpaths1}
If the $k$th shortest path, $p^{(k)}_{s,t}$, contains any node, $v$, such that $\forall i < k, v \notin p^{(i)}_{s,t}$, then $p^{(k)}_{s,t}$ is a via-path.
\end{lem}
\begin{proof}
Because $\forall i < k, v \notin p^{(i)}_{s,t}$ and $\forall i > k, p^{(i)}_{s,t} \ge p^{(k)}_{s,t}$, there can be no shorter path than $p^{(k)}_{s,t}$ that contains node $v$.
\end{proof}

\noindent The contrapositive of Lemma \ref{thm:kshortestpaths1} is that if the $k$th shortest path is not a via-path, then it involves only nodes present in the $(k-1)$ shortest paths.

\begin{lem}\label{thm:kshortestpaths3}
All paths shorter than a given CVP, $\tilde{p}_{s,v,t}$, include just those nodes in the set of CVPs shorter than $\tilde{p}_{s,v,t}$.
\end{lem}
\begin{proof}
By Theorem \ref{thm:kshortestpaths1}, all paths shorter than a path, $p$, include only those nodes in the set of via-paths shorter than $p$.
Since each node participates in a cascading via-path with the same cost as any via-path it participates in, the set of cascading via-paths with a cost below a particular cost, $c$, contains the same nodes as the set of all via-paths with cost below $c$.
\end{proof}

\begin{lem}\label{thm:shorter_rpcs}
All paths shorter than a given CVP, $\tilde{p}_{s,v,t}$, include just those nodes in the set of RPCs associated with the CVPs shorter than $\tilde{p}_{s,v,t}$.
\end{lem}
\begin{proof}
By Theorem \ref{thm:rpc_uniqueness}, every node, $u$, belongs to some maximal RPC. By Theorem \ref{thm:rpc}, this RPC is associated with a CVP $\tilde{p}_{s,u,t}$, with cost $c(\tilde{p}_{s,u,t})$. If $u$ is in a path shorter than $\tilde{p}_{s,v,t}$, then it must be the case that $c(\tilde{p}_{s,u,t}) < c(\tilde{p}_{s,v,t})$. Thus all nodes in paths shorter than $\tilde{p}_{s,v,t}$ are in an RPC associated with a CVP shorter than $\tilde{p}_{s,v,t}$.
\end{proof}

Let us define a graph $G^\prime = (V^\prime,E^\prime)$ such that $V^\prime$ contains every node in the RPC segments of the $k$ shortest CVPs, and $E^\prime = \{(u,v) | u \in V^\prime \land v \in V^\prime \land (u,v) \in E\}$.

\begin{thm}\label{thm:ksp_gprime}
Assuming $c(p_{s,t}^{(k)}) \ne c(p_{s,t}^{(k+1)})$---in other words, the $k$ shortest paths of $G$ are unique---the $k$ shortest paths of $G^\prime = (V^\prime,E^\prime)$ are the same as the $k$ shortest paths of $G$.
\end{thm}
\begin{proof}
A subgraph of graph $G$ cannot contain any paths that are not in $G$. 
Therefore, the $k$th shortest path of a subgraph of $G$ cannot be shorter than the $k$th shortest path of $G$.
Let $\tilde{p}_{s,v,t}^{(k)}$ be the $k$th shortest CVP in $G$.
There are at least $k$ shorter paths than $\tilde{p}_{s,v,t}^{(k+1)}$ in $G$.
By Lemma \ref{thm:shorter_rpcs}, the $k$ shortest paths of $G$ contain no nodes other than those in the RPCs associated with the $k$ shortest CVPs.
Node set $V^\prime$ contains all of the nodes in the RPCs associated with the $k$ shortest CVPs.
Edge set $E^\prime$ contains all of the edges in $E$ between nodes in $V^\prime$.
Therefore, any path involving only nodes in the RPCs associated with the $k$ shortest CVPs of $G$ is in $G^\prime$.
Thus, the $k$ shortest paths of $G$ are also in $G^\prime$, and must be the $k$ shortest paths of $G^\prime$. 
\end{proof}

\noindent It is not difficult to show that if the $k$ shortest paths of $G$ is not a unique set, due to $c(p_{s,t}^{(k)}) = c(p_{s,t}^{(k+1)})$, then the $k$ shortest paths of $G^\prime$ simply exhibit the same non-uniqueness.

While it is clear that computing the $k$ shortest paths of a subgraph $G^\prime$ such that $|V^\prime| \ll |V|$ and $|E^\prime| \ll |E|$ takes much less time than computing the $k$ shortest paths of $G$, it needs to be determined if this reduction is enough to offset the time required to compute $G^\prime$.
Unfortunately, this is a question that requires more space to fully answer than we can devote in this paper.
$G^\prime$ can be computed in $O(|E| + |V| \log |V|)$ time.
The lazy version of Eppstein's algorithm \cite{jimenez2003lazy} for computing the $k$ shortest paths operates in $O(|E| + |V| \log |V| + k \log k)$.
Because both algorithms have the same asymptotic dependency on the size of $G$, it is not clear from asymptotic analysis alone that computing and substituting $G^\prime$ would yield a net reduction in the time it takes to compute the $k$ shortest paths of $G$.

However, computing $G^\prime$ could yield a significant reduction in the overall time required to compute the $k$ shortest \emph{loopless} paths of $G$.
Yen's algorithm \cite{yen1971finding} for computing the $k$ shortest loopless paths operates in $O(k |V| (|E| + |V| \log |V|)$ time.
Computing $G^\prime$ requires an order of magnitude less computation than Yen's algorithm.
Additionally, a reduction in the sizes of the sets of nodes and edges of $G$ by a factor of $N$ yields approximately an $N^2$ reduction in computation for Yen's algorithm.
We should note that because CVPs are not guaranteed to be loopless, computing the RPCs of the $k$ shortest CVPs no longer guarantees $k$ loopless paths in $G^\prime$.
This can be accounted for by computing the RPCs of the $j \ge k$ shortest CVPs, where the $j$ shortest CVPs include $k$ loopless CVPs.
Even though it is not clear how much smaller $G^\prime$ can be made compared to $G$ without considering the specific characteristics of $G$ and the value of $k$, overall this analysis suggests the use of RPCs has the potential to substantially reduce the time required to compute the $k$ shortest loopless paths.

\subsection{Locally optimal paths}\label{sec:vnf_theory}

%

A common motivation for computing multiple paths between two nodes is to generate a set of \emph{candidates} for solving a single best-path problem, when much of the optimization criteria is unknown or unspecified.
However, typical objectives for computing multiple paths, such as the $k$ shortest paths, may not be best-suited to efficiently optimizing a diverse candidate set in cases such as these.
A conceptually different approach that may be more useful is to compute a set of \emph{locally optimal} paths, according to the known criteria.
In general, a solution is considered locally optimal if it is optimal within a small, continuous \emph{neighborhood} surrounding it within the space of all solutions.
Reducing the set of all paths to just those that are locally optimal eliminates potentially redundant candidates---within clusters of similar paths, only the best path according to the specified criteria is chosen to be a candidate.

To specify a concrete notion of similarity for paths, we use the Jaccard similarity,
\begin{equation}
J(p,q) = \frac{|p \cap q|}{|p \cup q|}.
\end{equation}
\noindent From this, the Jaccard distance, which we denote $\delta(p,q)$, is defined as,
\begin{equation}
\delta(p,q) = 1 - J(p,q).
\end{equation}

\noindent This allows us to define a space of paths in which to look for locally optimal solutions.
In the context of paths, however, it is not clear how to precisely define local optimality because the space of solutions is discrete rather than continuous, consisting of isolated points.
Shiller, et al. \cite{shiller2004computing} consider a path to be locally optimal if any small perturbation to the path maintains or increases its length/cost.
However, this definition is not of practical use without defining what constitutes a \emph{small} perturbation for a path.

The primary problem is that the discrete nature of paths suggests an arbitrary cut-off for what is and is not locally optimal.
Instead, we would like to realize a fuzzier notion of local optimality that is real-valued, defined on the interval $[0,1]$, indicating the \emph{degree} to which a path is locally optimal.
We propose to accomplish this by measuring the \emph{size} of the neighborhood, $\mathcal{N}(p)$, around path $p$ in the space of all paths, in which $p$ is the optimal path.

\begin{align}\label{eq:neighborhood}
\mathcal{N}(p) = \min_q \ \ &\delta(p,q),\\
\text{subject to} \ \ &c(p) > c(q). \nonumber
\end{align}
\noindent We denote the closest such path as,
\begin{align}\label{eq:neighborhood2}
\hat{q} = \argmin_q \ \ &\delta(p,q),\\
\text{subject to} \ \ &c(p) > c(q). \nonumber
\end{align}

\noindent For each path, $p$, the smallest possible value for $\mathcal{N}(p)$ is the distance to its closest neighbor, so $\mathcal{N}(p) > 0$. Additionally, $\mathcal{N}(p) \le 1$, with equality when $p$ is the shortest path in the graph.

This definition of local optimality gives us a framework for choosing the $k$ most locally optimal paths, which as we stated can be useful for generating candidate solutions in under-specified path optimization problems.
However, computing the Jaccard distance between every pair of paths can be computationally expensive. 
In this case, the following theorems provide an alternative framework that can be implemented much more efficiently.

\begin{thm}[Path similarity bound theorem]\label{thm:similarity_bound}
For paths $p$ and $q$, if $c(p) > c(q)$, then $\delta(p,q) \ge \omega(p)$.
\end{thm}
\begin{proof}
The via-node fraction, $\omega(p)$, has the form $\frac{|r_p|}{|p|}$, where $r_p$ is the sub-path of the RPC belonging to path $p$. Therefore, to show that $\delta(p,q) \ge \omega(p)$, we need to show that, $1 - \frac{|p \cap q|}{|p \cup q|} \ge \frac{|r_{p}|}{|p|}$. Because $|p| \le |p \cup q|$, it follows that $1 - \frac{|p \cap q|}{|p \cup q|} \ge 1 - \frac{|p \cap q|}{|p|}$. Thus we can prove the theorem by proving instead that $1 - \frac{|p \cap q|}{|p|} \ge \frac{|r_{p}|}{|p|}$. Multiplying both sides by $|p|$ and rearranging, we get $|p| \ge |r_{p}| + |p \cap q|$. Because $r_{p} \subseteq p$, we can prove this is true if we can show that $|r_{p} \cap q| = 0$. For any element $w \in r_{p}$, $c(\tilde{p}_{s,w,t}) = c(p)$. But since $c(q) < c(p)$, we know that $w \notin q$. This implies that $|r_{p} \cap q| = 0$, which proves the theorem.
\end{proof}

Theorem \ref{thm:similarity_bound} says that a path's via-node fraction places an upper limit on how similar a shorter path can be to it, according to Jaccard similarity. A corollary of this is the following:

\begin{cor}\label{thm:neighborhood}
$\mathcal{N}(p) \ge \omega(p)$.
\end{cor}

This provides a lower bound on the local optimality of each path, based on via-node fraction, $\omega(p)$, that is much more efficient to compute because it does not require distance comparisons between every pair of paths.
Since only CVPs have $\omega(p) > 0$, the maximal minimum local optimality is achieved with this bound by selecting the $k$ CVPs with the greatest via-node fraction. 

\subsection{Path set diversity}

In the previous section, we defined local optimality as a measurable quantity, and considered it as an objective for producing a set of paths.
In this section, we consider a related objective to this---the \emph{diversity} of the set of paths as a whole.

\begin{equation}\label{eq:diversity}
\eta(\mathcal{P}) = \frac{2}{|\mathcal{P}|(|\mathcal{P}| - 1)}\sum_{i < j}^{} \delta(\mathcal{P}_i,\mathcal{P}_j),
\end{equation}
\noindent which is the mean Jaccard distance between all pairs of paths in the set. Note that because $0 \le \delta(\mathcal{P}_i,\mathcal{P}_j) \le 1$, we have that $0 \le \eta(\mathcal{P}) \le 1$, as well.

One way to select a set of good paths with a large amount of diversity is to choose the $k$ shortest disjoint paths, for which $\eta(\mathcal{P})$ is maximized, with $\eta(\mathcal{P}) = 1$. However, optimizing this objective is computationally complex and in many applications, disjointness can be too severe of a restriction. Instead, we show that using CVPs, it is possible to simultaneously favor short paths and diversity without this restriction.

\begin{lem}\label{lem:cvp_distance_bound}
For CVPs, $p$ and $q$, $\delta(p,q) \ge \min(\omega(p),\omega(q))$.
\end{lem}
\begin{proof}
The via-node fraction, $\omega(p)$, has the form $\frac{|r_p|}{|p|}$, where $r_p$ is the sub-path of $p$ corresponding to its RPC. By Theorem \ref{thm:mrpc_disjointness}, because $r_p \ne r_q$, $|r_p \cap r_q| = 0$. From this, it follows that either $|r_p \cap q| = 0$ or $|r_q \cap p| = 0$. In the proof for Theorem \ref{thm:similarity_bound}, we proved that $\delta(p,q) \ge \omega(p)$ if $|r_{p} \cap q| = 0$. Since either $|r_p \cap q| = 0$ or $|r_q \cap p| = 0$, it follows that either $\delta(p,q) \ge \omega(p)$ or $\delta(p,q) \ge \omega(q)$, which implies that $\delta(p,q) \ge \min(\omega(p),\omega(q))$.
\end{proof}

\begin{thm}\label{thm:diversity_bound}
If $\breve{\Pi}$ is an ordered set of CVPs, sorted in descending order by via-node fraction, then
\begin{equation}\label{eq:diversity_bound}
\eta(\breve{\Pi}) \ge \frac{2}{|\breve{\Pi}|(|\breve{\Pi}| - 1)}\sum_{j=2}^{|\breve{\Pi}|} (j-1) \omega(\breve{\Pi}_j),
\end{equation}
\end{thm}
\begin{proof}
By Lemma \ref{lem:cvp_distance_bound},
\begin{equation}\label{eq:diversity_bound2}
\eta(\breve{\Pi}) \ge \frac{2}{|\breve{\Pi}|(|\breve{\Pi}| - 1)}\sum_{j=2}^{|\breve{\Pi}|}\sum_{i=1}^{j-1} \min(\omega(\breve{\Pi}_{i}),\omega(\breve{\Pi}_{j})).
\end{equation} 
\noindent Because $\breve{\Pi}$ is sorted in descending order by via-node fraction, $\omega(\breve{\Pi}_{j-1}) \ge \omega(\breve{\Pi}_{j})$. Therefore,
\begin{equation}\label{eq:diversity_bound2a}
\eta(\breve{\Pi}) \ge \frac{2}{|\breve{\Pi}|(|\breve{\Pi}| - 1)}\sum_{j=2}^{|\breve{\Pi}|}\sum_{i=1}^{j-1} \omega(\breve{\Pi}_{j}).
\end{equation} 
\noindent Since the term in the inner sum does not depend on index $i$,
\begin{equation}\label{eq:diversity_bound3}
\sum_{i=1}^{j-1} \omega(\hat{\Pi}_j) = (j-1) \omega(\hat{\Pi}_j).
\end{equation} 
\noindent Substituting Equation \ref{eq:diversity_bound3} into Equation \ref{eq:diversity_bound2}, yields Equation \ref{eq:diversity_bound}.
\end{proof}

One simple way to use Theorem \ref{thm:diversity_bound} to create a set of $k$ paths that is simultaneously short and diverse is to first select the $l > k$ CVPs with the greatest via-node fraction.
Any subset of $k$ paths from this set must have a diversity at least as great as that of the $k$-of-$l$ CVPs with the smallest via-node fraction within this set.
Then, with a minimum guaranteed diversity, one could choose from among the $l$ CVPs, the $k$ shortest CVPs.
This approach, which simultaneously optimizes diversity and length, is more efficient than choosing the $k$ shortest disjoint paths and does not impose a potentially severe constraint such as disjointness.

%

\section{Applications}\label{sec:results}

In this section, we demonstrate the RPC method in practice by applying it to two different problems: alternative route finding in road networks and layer-boundary identification in ground-penetrating radar (GPR) data.
These applications offer insight into the types of problems to which the RPC method can be applied and how the theory can serve as a principled foundation for solving them.

\subsection{Alternative route finding in road networks}\label{sec:alternative_routes}

\begin{table}
\centering
\small
    \begin{tabular}{r p{6.5cm} r}
    \hline
    Types & Description & Speed\\ \hline
    11--15 & Primary road with limited access & 70 mph\\
    25 & Primary road without limited access (separated) & 60 mph\\
    21--24 & Primary road without limited access (unseparated) & 55 mph\\
    31--35, 38 & Secondary and connecting road & 37.5 mph\\
    41--48 & Local, neighborhood, and rural road & 22.5 mph\\
    \hline
    \end{tabular}
  \caption{The different categories of road within the Florida road network dataset and the speed we assigned to each, capturing the presumed average speed of a typical driver in congestion-free conditions, taking the presence of traffic lights and signs into account.}
  \label{tab:speed_of_travel}
\end{table}

The first application we present is the problem of alternative route finding in road networks \cite{abraham2013alternative}. 
This is an example of a candidate-set generation problem, in which a single optimal solution is desired but the optimality of the solution is subject to criteria that are either unknown or under-specified at computation time.
The goal in such a problem is to compute a set of candidates that are good according to the specified criteria, yet diverse, from which a selection can be made later based on additional criteria.
This problem arises in automated driving route suggestion by navigation assistance systems because while these systems offer fast computation of route suggestions according to objective criteria, they are often unable to account for the particular set of unknown or changing preferences that a driver may have in choosing a best route.

 We present an experiment using road network data for the state of Florida, United States \cite{demetrescu20069th} to compute a set of alternative driving routes from the city of Tampa to the city of Miami. 
The dataset consists of a weighted, undirected graph, containing 1,048,506 nodes and 1,330,551 edges.
Each edge represents a segment of road and each edge weight indicates the length (distance) of the road segment.
Road segments in this dataset are categorized into 5 different high-level categories and 24 different low-level categories, by their characteristics and purpose.

Although the dataset specifies physical distance, we want to use travel time to optimize the alternative route suggestion.
To do so, we use the fact that the dataset contains information about the types of roads involved to estimate an expected speed of travel across each segment (edge), based on its road type.
The speeds we assign to each road type are listed in Table \ref{tab:speed_of_travel}.
From these speeds, we are able to compute an expected time of travel for each edge in the graph.

To compute the set of alternative routes, we use the four-stage method outlined in Section \ref{sec:summary}.
Instead of determining an absolute ranking for the CVPs using some measure and selecting the top $k$ paths, we specify two separate criteria and a threshold for each, and select each path satisfying these two thresholds.
In terms of path properties, these criteria are length (travel time) and RPC cost fraction.
The motivation for using travel time is to ensure that each alternative route is not significantly longer than the single best route, while the motivation for using RPC cost fraction, which can be interpreted as a node-weighted version of via-node fraction, is to encourage local optimality and diversity among the paths.
Automatically learning good values for these thresholds is out of the scope of this paper.
Instead we adjust them manually so that they balance path shortness and diversity in the resulting set.
For travel time, we specified the threshold in terms of the ratio of the length of each path divided by the length of the shortest path from Tampa to Miami, and set an upper bound for this of 1.33.
For RPC cost fraction, which for a path can vary between 0 and 1, we set a lower threshold of 0.175.

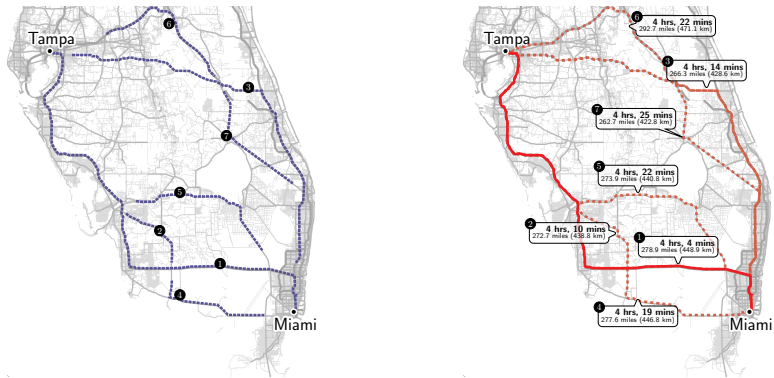
\begin{figure}[t]
	 \centering
	 \begin{subfigure}[b]{0.49\textwidth}
	 \centering
	     \resizebox{0.8\textwidth}{!}{
			 \begin{tikzpicture}
			 \begin{axis}[ymin=25.150396,ymax= 28.3475538,xmin=-83.051126,xmax=-79.833118,
				 height=14.8cm,width=13.7cm,ticks=none,hide axis]
			 \addplot graphics [ymin=25.250396,ymax= 28.8475538,xmin=-82.851126,xmax=-80.033118] {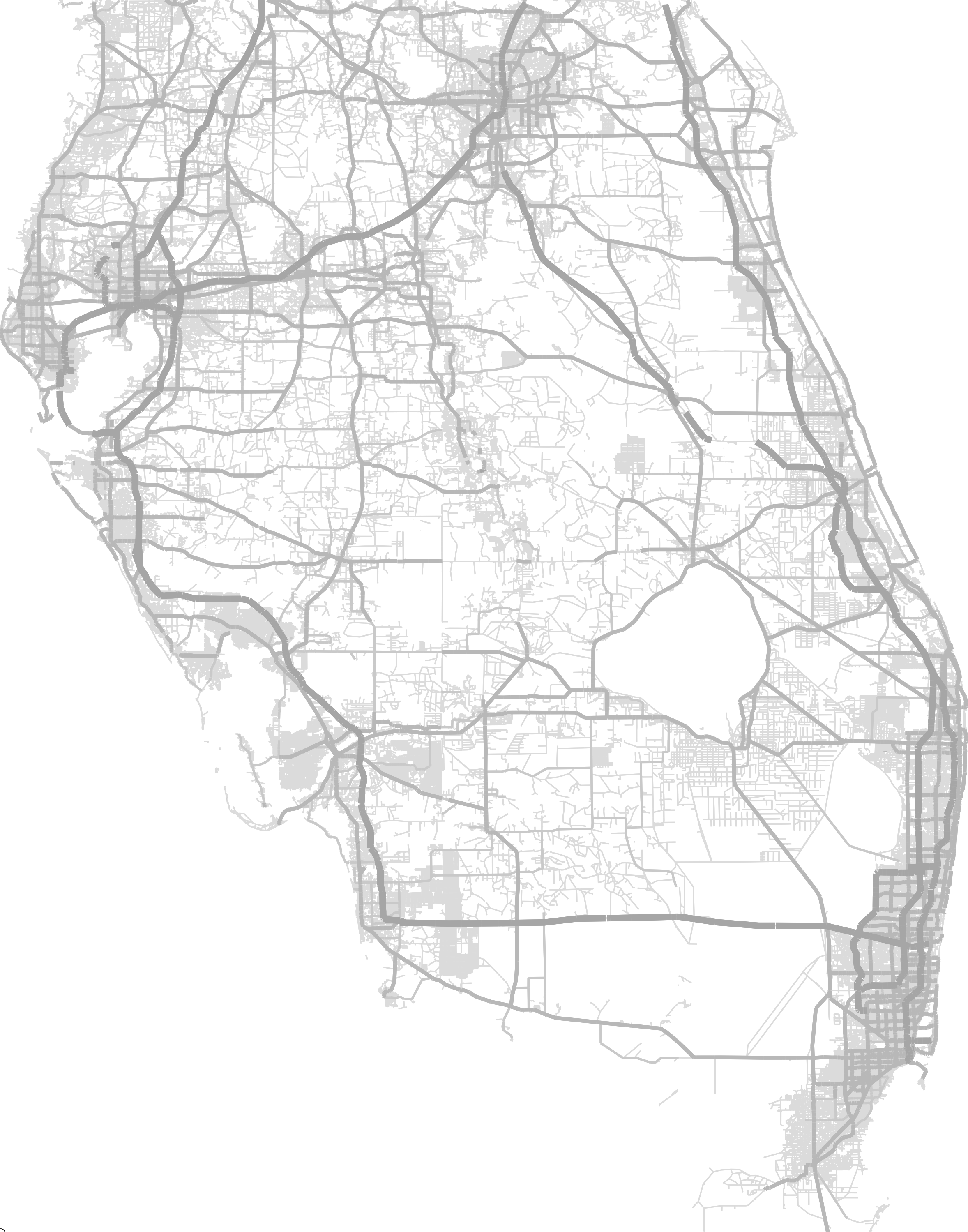};
	 	 	 \addplot[mark=*,mark size=3pt,color=white] coordinates {(-82.465, 27.971)};
		 	 \addplot[mark=*,mark size=1.5pt,color=black] coordinates {(-82.465, 27.971)};
    			 \node [] at (axis cs:-82.465, 28.071) {\contourlength{2pt} \contour{white}{\LARGE \textsf{Tampa}}};
	 		 \addplot[mark=*,mark size=3pt,color=white] coordinates {(-80.2241, 25.7877)};
			 \addplot[mark=*,mark size=1.5pt,color=black] coordinates {(-80.2241, 25.7877)};
			 \node [] at (axis cs:-80.2241, 25.6877) {\contourlength{2pt} \contour{white}{\LARGE \textsf{Miami}}};	
				 \addplot[color=Blue!70!Gray,line width=2.5,dash pattern=on 3pt off 1pt]
				 table[x index=1,y index=2] {RPCs_Tampa_to_Miami_time_7.txt};
				 \addplot[color=Blue!70!Gray,line width=2.5,dash pattern=on 3pt off 1pt]
				 table[x index=1,y index=2] {RPCs_Tampa_to_Miami_time_6.txt};
				 \addplot[color=Blue!70!Gray,line width=2.5,dash pattern=on 3pt off 1pt]
				 table[x index=1,y index=2] {RPCs_Tampa_to_Miami_time_5.txt};
				 \addplot[color=Blue!70!Gray,line width=2.5,dash pattern=on 3pt off 1pt]
				 table[x index=1,y index=2] {RPCs_Tampa_to_Miami_time_4.txt};
				 \addplot[color=Blue!70!Gray,line width=2.5,dash pattern=on 3pt off 1pt]
				 table[x index=1,y index=2] {RPCs_Tampa_to_Miami_time_3.txt};
				 \addplot[color=Blue!70!Gray,line width=2.5,dash pattern=on 3pt off 1pt]
				 table[x index=1,y index=2] {RPCs_Tampa_to_Miami_time_2.txt};
				 \addplot[color=Blue!70!Gray,line width=2.5,dash pattern=on 3pt off 1pt]
				 table[x index=1,y index=2] {RPCs_Tampa_to_Miami_time_1.txt};
				 \node[] at (axis cs:-80.879139, 26.172387) (b0) {};
					 \node[circle,draw=black,inner sep=1pt,fill=black,fill opacity=1,text=white,text opacity=1,align=right,callout absolute pointer=(b0),above left=-5pt and -5pt of b0] {\small{1}};					 \node[] at (axis cs:-81.433914, 26.448313) (b1) {};
					 \node[circle,draw=black,inner sep=1pt,fill=black,fill opacity=1,text=white,text opacity=1,align=right,callout absolute pointer=(b1),above left=-5pt and -5pt of b1] {\small{2}};					 \node[] at (axis cs:-80.618461, 27.647494) (b2) {};
					 \node[circle,draw=black,inner sep=1pt,fill=black,fill opacity=1,text=white,text opacity=1,align=right,callout absolute pointer=(b2),above left=-5pt and -5pt of b2] {\small{3}};					  \node[] at (axis cs:-81.248225, 25.910843) (b3) {};
					 \node[circle,draw=black,inner sep=1pt,fill=black,fill opacity=1,text=white,text opacity=1,align=right,callout absolute pointer=(b3),above left=-5pt and -5pt of b3] {\small{4}};					 \node[] at (axis cs:-81.248225, 26.768397) (b4) {};
					 \node[circle,draw=black,inner sep=1pt,fill=black,fill opacity=1,text=white,text opacity=1,align=right,callout absolute pointer=(b4),above left=-5pt and -5pt of b4] {\small{5}};				 	\node[] at (axis cs:-81.348225, 28.181889) (b5) {};
					 \node[circle,draw=black,inner sep=1pt,fill=black,fill opacity=1,text=white,text opacity=1,align=right,callout absolute pointer=(b5),above left=-5pt and -5pt of b5] {\small{6}};				 	\node[] at (axis cs:-80.820697, 27.243747) (b6) {};
					 \node[circle,draw=black,inner sep=1pt,fill=black,fill opacity=1,text=white,text opacity=1,align=right,callout absolute pointer=(b6),above left=-5pt and -5pt of b6] {\small{7}};
				 \end{axis}
			 \end{tikzpicture}
		}
		 \caption{Reciprocal pointer chains.}
		 \label{sfig:tampa_miami_rpcs}
	 \end{subfigure}
	 \begin{subfigure}[b]{0.49\textwidth}
	 \centering
	     \resizebox{0.8\textwidth}{!}{
		 \begin{tikzpicture}
			 \begin{axis}[ymin=25.150396,ymax= 28.3475538,xmin=-83.051126,xmax=-79.833118,
				 height=14.8cm,width=13.7cm,ticks=none,hide axis]
			 \addplot graphics [ymin=25.250396,ymax= 28.8475538,xmin=-82.851126,xmax=-80.033118] {tampa_miami_background_map_01.png};
	 	 	 \addplot[mark=*,mark size=3pt,color=white] coordinates {(-82.465, 27.971)};
		 	 \addplot[mark=*,mark size=1.5pt,color=black] coordinates {(-82.465, 27.971)};
    			 \node [] at (axis cs:-82.465, 28.071) {\contourlength{2pt} \contour{white}{\LARGE \textsf{Tampa}}};
	 		 \addplot[mark=*,mark size=3pt,color=white] coordinates {(-80.2241, 25.7877)};
			 \addplot[mark=*,mark size=1.5pt,color=black] coordinates {(-80.2241, 25.7877)};
			 \node [] at (axis cs:-80.2241, 25.6877) {\contourlength{2pt} \contour{white}{\LARGE \textsf{Miami}}};	
				 \addplot[color=Red!70!Gray,line width=2.5,dash pattern=on 3pt off 3pt]
				 table[x index=1,y index=2] {explicit_via_paths_Tampa_to_Miami_time_7.txt};
				 \addplot[color=Red!70!Gray,line width=2.5,dash pattern=on 3pt off 3pt]
				 table[x index=1,y index=2] {explicit_via_paths_Tampa_to_Miami_time_6.txt};
				 \addplot[color=Red!70!Gray,line width=2.5,dash pattern=on 3pt off 3pt]
				 table[x index=1,y index=2] {explicit_via_paths_Tampa_to_Miami_time_5.txt};
				 \addplot[color=Red!70!Gray,line width=2.5,dash pattern=on 3pt off 3pt]
				 table[x index=1,y index=2] {explicit_via_paths_Tampa_to_Miami_time_4.txt};
				 \addplot[color=Red!70!Gray,line width=2.5,dash pattern=on 3pt off 3pt]
				 table[x index=1,y index=2] {explicit_via_paths_Tampa_to_Miami_time_3.txt};
				 \addplot[color=Red!70!Gray,line width=2.5,dash pattern=on 3pt off 3pt]
				 table[x index=1,y index=2] {explicit_via_paths_Tampa_to_Miami_time_2.txt};
				 \addplot[color=Red,line width=2.5]
				 table[x index=1,y index=2] {explicit_via_paths_Tampa_to_Miami_time_1.txt};
				 \node[] at (axis cs:-80.879139, 26.172387) (h0) {};
				 \node[rectangle callout,rounded corners,draw=black,inner sep=3pt,fill=white,fill opacity=0.7,text opacity=1,align=right,callout absolute pointer=(h0),above=5pt of h0] (b0)
					 {\small{\textbf{\textsf{4 hrs, 4 mins}}} \\[-2mm] \scriptsize{\textsf{278.9 miles (448.9 km)}}};
					 \node[circle,draw=black,inner sep=1pt,fill=black,fill opacity=1,text=white,text opacity=1,align=right,callout absolute pointer=(b0),above left=-5pt and -5pt of b0] {\small{1}};					 \node[] at (axis cs:-81.433914, 26.448313) (h1) {};
				 \node[rectangle callout,rounded corners,draw=black,inner sep=3pt,fill=white,fill opacity=0.7,text opacity=1,align=right,callout absolute pointer=(h1),left=5pt of h1] (b1)
					 {\small{\textbf{\textsf{4 hrs, 10 mins}}} \\[-2mm] \scriptsize{\textsf{272.7 miles (438.8 km)}}};
					 \node[circle,draw=black,inner sep=1pt,fill=black,fill opacity=1,text=white,text opacity=1,align=right,callout absolute pointer=(b1),above left=-5pt and -5pt of b1] {\small{2}};					 \node[] at (axis cs:-80.618461, 27.647494) (h2) {};
				 \node[rectangle callout,rounded corners,draw=black,inner sep=3pt,fill=white,fill opacity=0.7,text opacity=1,align=right,callout absolute pointer=(h2),above=5pt of h2] (b2)
					 {\small{\textbf{\textsf{4 hrs, 14 mins}}} \\[-2mm] \scriptsize{\textsf{266.3 miles (428.6 km)}}};
					 \node[circle,draw=black,inner sep=1pt,fill=black,fill opacity=1,text=white,text opacity=1,align=right,callout absolute pointer=(b2),above left=-5pt and -5pt of b2] {\small{3}};					  \node[] at (axis cs:-81.248225, 25.910843) (h3) {};
				 \node[rectangle callout,rounded corners,draw=black,inner sep=3pt,fill=white,fill opacity=0.7,text opacity=1,align=right,callout absolute pointer=(h3),below=5pt of h3] (b3)
					 {\small{\textbf{\textsf{4 hrs, 19 mins}}} \\[-2mm] \scriptsize{\textsf{277.6 miles (446.8 km)}}};
					 \node[circle,draw=black,inner sep=1pt,fill=black,fill opacity=1,text=white,text opacity=1,align=right,callout absolute pointer=(b3),above left=-5pt and -5pt of b3] {\small{4}};					 \node[] at (axis cs:-81.248225, 26.768397) (h4) {};
				 \node[rectangle callout,rounded corners,draw=black,inner sep=3pt,fill=white,fill opacity=0.7,text opacity=1,align=right,callout absolute pointer=(h4),above=5pt of h4] (b4)
					 {\small{\textbf{\textsf{4 hrs, 22 mins}}} \\[-2mm] \scriptsize{\textsf{273.9 miles (440.8 km)}}};
					 \node[circle,draw=black,inner sep=1pt,fill=black,fill opacity=1,text=white,text opacity=1,align=right,callout absolute pointer=(b4),above left=-5pt and -5pt of b4] {\small{5}};				 	\node[] at (axis cs:-81.348225, 28.181889) (h5) {};
				 \node[rectangle callout,rounded corners,draw=black,inner sep=3pt,fill=white,fill opacity=0.7,text opacity=1,align=right,callout absolute pointer=(h5),right=5pt of h5] (b5)
					 {\small{\textbf{\textsf{4 hrs, 22 mins}}} \\[-2mm] \scriptsize{\textsf{292.7 miles (471.1 km)}}};
					 \node[circle,draw=black,inner sep=1pt,fill=black,fill opacity=1,text=white,text opacity=1,align=right,callout absolute pointer=(b5),above left=-5pt and -5pt of b5] {\small{6}};				 	\node[] at (axis cs:-80.820697, 27.243747) (h6) {};
				 \node[rectangle callout,rounded corners,draw=black,inner sep=3pt,fill=white,fill opacity=0.7,text opacity=1,align=right,callout absolute pointer=(h6),above left=5pt and 5pt of h6] (b6)
					 {\small{\textbf{\textsf{4 hrs, 25 mins}}} \\[-2mm] \scriptsize{\textsf{262.7 miles (422.8 km)}}};
					 \node[circle,draw=black,inner sep=1pt,fill=black,fill opacity=1,text=white,text opacity=1,align=right,callout absolute pointer=(b6),above left=-5pt and -5pt of b6] {\small{7}};
			 \end{axis}
		 \end{tikzpicture}
		 } 
		 \caption{Shortest route and alternative routes.}
		 \label{sfig:tampa_miami_alternative_routes}
	\end{subfigure}
	\caption{Figure \ref{sfig:tampa_miami_rpcs} shows seven RPCs computed for the Florida road network with Tampa as the source and Miami as the target. Figure \ref{sfig:tampa_miami_alternative_routes} shows the corresponding CVPs, which for this application serve as alternative routes. The shortest route is depicted with a solid line, while the alternative routes are dashed. Because the full routes for some of these paths may be ambiguous, they are depicted individually in Figure \ref{fig:tampa_miami_individual_routes}.}
	\label{fig:tampa_miami_routes}
\end{figure}

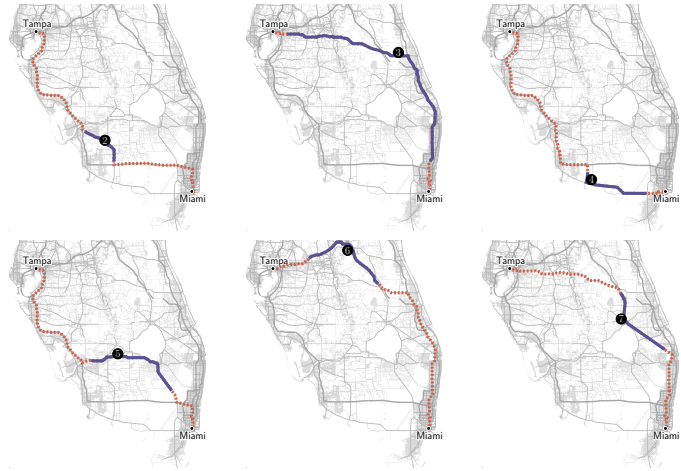
\begin{figure}[t]
	 \centering
	    \begin{subfigure}[b]{0.25\textwidth}
              \centering
              \resizebox{1.0\textwidth}{!}{
		 \begin{tikzpicture}
			 \begin{axis}[ymin=25.150396,ymax= 28.3475538,xmin=-83.051126,xmax=-79.833118,
				 height=14.8cm,width=14.2cm,ticks=none,hide axis]
			 \addplot graphics [ymin=25.250396,ymax= 28.8475538,xmin=-82.851126,xmax=-80.033118] {tampa_miami_background_map_01.png};
			 \addplot[mark=*,mark size=3pt,color=white] coordinates {(-82.465, 27.971)};
		 	 \addplot[mark=*,mark size=1.5pt,color=black] coordinates {(-82.465, 27.971)};
    			 \node [] at (axis cs:-82.465, 28.071) {\contourlength{2pt} \contour{white}{\LARGE \textsf{Tampa}}};
	 		 \addplot[mark=*,mark size=3pt,color=white] coordinates {(-80.2241, 25.7877)};
			 \addplot[mark=*,mark size=1.5pt,color=black] coordinates {(-80.2241, 25.7877)};
			 \node [] at (axis cs:-80.2241, 25.6877) {\contourlength{2pt} \contour{white}{\LARGE \textsf{Miami}}};	
				 \addplot[color=Red!70!Gray,line width=5,dash pattern=on 3pt off 3pt]
				 table[x index=1,y index=2] {explicit_via_paths_Tampa_to_Miami_time_2.txt};
				 \addplot[color=Blue!70!Gray,line width=5]
				 table[x index=1,y index=2] {RPCs_Tampa_to_Miami_time_2.txt};
				 \node[] at (axis cs:-81.433914, 26.448313) (b1) {};
					 \node[circle,draw=black,inner sep=1pt,fill=black,fill opacity=1,text=white,text opacity=1,align=right,callout absolute pointer=(b1),above left=-5pt and -5pt of b1] {\LARGE{2}};	
			 \end{axis}
		 \end{tikzpicture}
	      }
	    \end{subfigure}
	    \begin{subfigure}[b]{0.25\textwidth}
              \centering
              \resizebox{1.0\textwidth}{!}{
		 \begin{tikzpicture}
			 \begin{axis}[ymin=25.150396,ymax= 28.3475538,xmin=-83.051126,xmax=-79.833118,
				 height=14.8cm,width=14.2cm,ticks=none,hide axis]
			 \addplot graphics [ymin=25.250396,ymax= 28.8475538,xmin=-82.851126,xmax=-80.033118] {tampa_miami_background_map_01.png};
			 \addplot[mark=*,mark size=3pt,color=white] coordinates {(-82.465, 27.971)};
		 	 \addplot[mark=*,mark size=1.5pt,color=black] coordinates {(-82.465, 27.971)};
    			 \node [] at (axis cs:-82.465, 28.071) {\contourlength{2pt} \contour{white}{\LARGE \textsf{Tampa}}};
	 		 \addplot[mark=*,mark size=3pt,color=white] coordinates {(-80.2241, 25.7877)};
			 \addplot[mark=*,mark size=1.5pt,color=black] coordinates {(-80.2241, 25.7877)};
			 \node [] at (axis cs:-80.2241, 25.6877) {\contourlength{2pt} \contour{white}{\LARGE \textsf{Miami}}};	
			 \addplot[color=Red!70!Gray,line width=5,dash pattern=on 3pt off 3pt]
				 table[x index=1,y index=2] {explicit_via_paths_Tampa_to_Miami_time_3.txt};
				 \addplot[color=Blue!70!Gray,line width=5]
				 table[x index=1,y index=2] {RPCs_Tampa_to_Miami_time_3.txt};
				 \node[] at (axis cs:-80.618461, 27.647494) (b2) {};
					 \node[circle,draw=black,inner sep=1pt,fill=black,fill opacity=1,text=white,text opacity=1,align=right,callout absolute pointer=(b2),above left=-5pt and -5pt of b2] {\LARGE{3}};	
			 \end{axis}
		 \end{tikzpicture}
	      }
	    \end{subfigure}
	    \begin{subfigure}[b]{0.25\textwidth}
              \centering
              \resizebox{1.0\textwidth}{!}{
		 \begin{tikzpicture}
			 \begin{axis}[ymin=25.150396,ymax= 28.3475538,xmin=-83.051126,xmax=-79.833118,
				 height=14.8cm,width=14.2cm,ticks=none,hide axis]
			 \addplot graphics [ymin=25.250396,ymax= 28.8475538,xmin=-82.851126,xmax=-80.033118] {tampa_miami_background_map_01.png};
			 \addplot[mark=*,mark size=3pt,color=white] coordinates {(-82.465, 27.971)};
		 	 \addplot[mark=*,mark size=1.5pt,color=black] coordinates {(-82.465, 27.971)};
    			 \node [] at (axis cs:-82.465, 28.071) {\contourlength{2pt} \contour{white}{\LARGE \textsf{Tampa}}};
	 		 \addplot[mark=*,mark size=3pt,color=white] coordinates {(-80.2241, 25.7877)};
			 \addplot[mark=*,mark size=1.5pt,color=black] coordinates {(-80.2241, 25.7877)};
			 \node [] at (axis cs:-80.2241, 25.6877) {\contourlength{2pt} \contour{white}{\LARGE \textsf{Miami}}};	
			 \addplot[color=Red!70!Gray,line width=5,dash pattern=on 3pt off 3pt]
				 table[x index=1,y index=2] {explicit_via_paths_Tampa_to_Miami_time_4.txt};
				 \addplot[color=Blue!70!Gray,line width=5]
				 table[x index=1,y index=2] {RPCs_Tampa_to_Miami_time_4.txt};
				  \node[] at (axis cs:-81.248225, 25.910843) (b3) {};
					 \node[circle,draw=black,inner sep=1pt,fill=black,fill opacity=1,text=white,text opacity=1,align=right,callout absolute pointer=(b3),above left=-5pt and -5pt of b3] {\LARGE{4}};	
			 \end{axis}
		 \end{tikzpicture}
	      }
	    \end{subfigure}
	    
	 \centering
	    \begin{subfigure}[b]{0.25\textwidth}
              \centering
              \resizebox{1.0\textwidth}{!}{
		 \begin{tikzpicture}
			 \begin{axis}[ymin=25.150396,ymax= 28.3475538,xmin=-83.051126,xmax=-79.833118,
				 height=14.8cm,width=14.2cm,ticks=none,hide axis]
			 \addplot graphics [ymin=25.250396,ymax= 28.8475538,xmin=-82.851126,xmax=-80.033118] {tampa_miami_background_map_01.png};
			 \addplot[mark=*,mark size=3pt,color=white] coordinates {(-82.465, 27.971)};
		 	 \addplot[mark=*,mark size=1.5pt,color=black] coordinates {(-82.465, 27.971)};
    			 \node [] at (axis cs:-82.465, 28.071) {\contourlength{2pt} \contour{white}{\LARGE \textsf{Tampa}}};
	 		 \addplot[mark=*,mark size=3pt,color=white] coordinates {(-80.2241, 25.7877)};
			 \addplot[mark=*,mark size=1.5pt,color=black] coordinates {(-80.2241, 25.7877)};
			 \node [] at (axis cs:-80.2241, 25.6877) {\contourlength{2pt} \contour{white}{\LARGE \textsf{Miami}}};	
			\addplot[color=Red!70!Gray,line width=5,dash pattern=on 3pt off 3pt]
				 table[x index=1,y index=2] {explicit_via_paths_Tampa_to_Miami_time_5.txt};
				 \addplot[color=Blue!70!Gray,line width=5]
				 table[x index=1,y index=2] {RPCs_Tampa_to_Miami_time_5.txt};
				 \node[] at (axis cs:-81.248225, 26.768397) (b4) {};
					 \node[circle,draw=black,inner sep=1pt,fill=black,fill opacity=1,text=white,text opacity=1,align=right,callout absolute pointer=(b4),above left=-5pt and -5pt of b4] {\LARGE{5}};
			 \end{axis}
		 \end{tikzpicture}
	      }
	    \end{subfigure}
	    \begin{subfigure}[b]{0.25\textwidth}
              \centering
              \resizebox{1.0\textwidth}{!}{
		 \begin{tikzpicture}
			 \begin{axis}[ymin=25.150396,ymax= 28.3475538,xmin=-83.051126,xmax=-79.833118,
				 height=14.8cm,width=14.2cm,ticks=none,hide axis]
			 \addplot graphics [ymin=25.250396,ymax= 28.8475538,xmin=-82.851126,xmax=-80.033118] {tampa_miami_background_map_01.png};
			 \addplot[mark=*,mark size=3pt,color=white] coordinates {(-82.465, 27.971)};
		 	 \addplot[mark=*,mark size=1.5pt,color=black] coordinates {(-82.465, 27.971)};
    			 \node [] at (axis cs:-82.465, 28.071) {\contourlength{2pt} \contour{white}{\LARGE \textsf{Tampa}}};
	 		 \addplot[mark=*,mark size=3pt,color=white] coordinates {(-80.2241, 25.7877)};
			 \addplot[mark=*,mark size=1.5pt,color=black] coordinates {(-80.2241, 25.7877)};
			 \node [] at (axis cs:-80.2241, 25.6877) {\contourlength{2pt} \contour{white}{\LARGE \textsf{Miami}}};	
			 \addplot[color=Red!70!Gray,line width=5,dash pattern=on 3pt off 3pt]
				 table[x index=1,y index=2] {explicit_via_paths_Tampa_to_Miami_time_6.txt};
				 \addplot[color=Blue!70!Gray,line width=5]
				 table[x index=1,y index=2] {RPCs_Tampa_to_Miami_time_6.txt};
			 	\node[] at (axis cs:-81.348225, 28.181889) (b5) {};
					 \node[circle,draw=black,inner sep=1pt,fill=black,fill opacity=1,text=white,text opacity=1,align=right,callout absolute pointer=(b5),above left=-5pt and -5pt of b5] {\LARGE{6}};
			 \end{axis}
		 \end{tikzpicture}
	      }
	    \end{subfigure}
	    \begin{subfigure}[b]{0.25\textwidth}
              \centering
              \resizebox{1.0\textwidth}{!}{
		 \begin{tikzpicture}
			 \begin{axis}[ymin=25.150396,ymax= 28.3475538,xmin=-83.051126,xmax=-79.833118,
				 height=14.8cm,width=14.2cm,ticks=none,hide axis]
			 \addplot graphics [ymin=25.250396,ymax= 28.8475538,xmin=-82.851126,xmax=-80.033118] {tampa_miami_background_map_01.png};
			 \addplot[mark=*,mark size=3pt,color=white] coordinates {(-82.465, 27.971)};
		 	 \addplot[mark=*,mark size=1.5pt,color=black] coordinates {(-82.465, 27.971)};
    			 \node [] at (axis cs:-82.465, 28.071) {\contourlength{2pt} \contour{white}{\LARGE \textsf{Tampa}}};
	 		 \addplot[mark=*,mark size=3pt,color=white] coordinates {(-80.2241, 25.7877)};
			 \addplot[mark=*,mark size=1.5pt,color=black] coordinates {(-80.2241, 25.7877)};
			 \node [] at (axis cs:-80.2241, 25.6877) {\contourlength{2pt} \contour{white}{\LARGE \textsf{Miami}}};	
			 \addplot[color=Red!70!Gray,line width=5,dash pattern=on 3pt off 3pt]
				 table[x index=1,y index=2] {explicit_via_paths_Tampa_to_Miami_time_7.txt};
				 \addplot[color=Blue!70!Gray,line width=5]
				 table[x index=1,y index=2] {RPCs_Tampa_to_Miami_time_7.txt};
			 	\node[] at (axis cs:-80.820697, 27.243747) (b6) {};
					 \node[circle,draw=black,inner sep=1pt,fill=black,fill opacity=1,text=white,text opacity=1,align=right,callout absolute pointer=(b6),above left=-5pt and -5pt of b6] {\LARGE{7}};
			 \end{axis}
		 \end{tikzpicture}
	      }
	    \end{subfigure}
	\caption{This figure shows the full alternative routes depicted in Figure \ref{fig:tampa_miami_routes}. The RPC portion of each route is illustrated in solid blue, while the rest of the full path is dashed red.}
	\label{fig:tampa_miami_individual_routes}
\end{figure}

Seven CVPs computed using the RPC method met the minimum thresholds for both normalized length and RPC cost fraction.
These are depicted in Figures \ref{fig:tampa_miami_routes} and \ref{fig:tampa_miami_individual_routes}.
Figure \ref{sfig:tampa_miami_rpcs} shows the RPC associated with each of the CVPs.
Figure \ref{sfig:tampa_miami_alternative_routes} shows the full CVPs overlaid on top of one another, along with time and distance information about each route.
Additional information about the CVPs is given in Table \ref{tab:tampa_miami_routes}.

Because it is not clear from Figure \ref{sfig:tampa_miami_alternative_routes} what some of the full routes are, we break these out individually and show them in full in Figure \ref{fig:tampa_miami_individual_routes}.
Importantly, the routes found by our method appear to be qualitatively reasonable and match those that drivers with expert knowledge of the local roads would consider.

\begin{table}
\centering
\small
    \begin{tabular}{R{1cm} R{2.5cm} R{2cm} R{2cm} R{2.3cm}}
    \hline
    Path \newline No. & Total \newline Time & Total \newline Time Rank & RPC Cost \newline Fraction & RPC Cost \newline Fraction Rank\\ \hline
    1 & 4 hrs, 4 mins & 1 & 1.0 & 1\\
    2 & 4 hrs, 10 mins & 2 & 0.20 & 7\\
    3 & 4 hrs, 14 mins & 3 & 0.83 & 2\\
    4 & 4 hrs, 19 mins & 4 & 0.27 & 6\\
    5 & 4 hrs, 22 mins & 5 & 0.38 & 3\\
    6 & 4 hrs, 22 mins & 6 & 0.3 & 5\\
    7 & 4 hrs, 25 mins & 7 & 0.3 & 4\\
    \hline
    \end{tabular}
  \caption{The $k = 7$ alternative routes, depicted in Figures \ref{fig:tampa_miami_routes} and \ref{fig:tampa_miami_individual_routes}, that met the specified thresholds of RPC cost fraction $\ge$ 0.175 and total cost $\le$ shortest path cost * 1.33. In this case, with the shortest path cost being 4 hrs, 4 mins, the threshold on total path cost was set at 5 hrs, 25 mins. All of the paths meeting both thresholds are well under the cost threshold.}
  \label{tab:tampa_miami_routes}
\end{table}

In general, the preceding result demonstrates the suitability of CVPs for generating a set of candidate solutions to an under-specified shortest path problem.
What is particularly notable about this result is that it was produced from new theory that allowed us to directly incorporate path diversity into the best-paths objective, through the use of an intrinsic and efficient-to-compute property of each individual CVP.
That each CVP contributes independently to this objective avoids a combinatorial optimization problem for finding the best set of paths, and the theoretical motivation provides a principled approach for creating a diverse set of alternatives.

\subsection{Layer-boundary identification in GPR data}\label{sec:layer_results}

The second problem we consider is the identification of boundaries between different soil layers in ground-penetrating radar (GPR) images.
Identifying layer-boundaries is an important problem in landmine detection using GPR.
A GPR system operates by transmitting an electromagnetic pulse downward into the ground and recording a return signal at a receiver.
When the pulse propagating through the air and ground reaches the boundary between two materials with different dielectric properties, a portion of it reflects back up to the receiver.
This allows GPR to sense buried objects whose dielectric properties differ from that of the surrounding soil.

\begin{figure}
    \centering
         \resizebox{0.6\textwidth}{!}{
            \begin{tikzpicture}
             \begin{axis}[
              xmin=0.5,xmax=24.5,ymin=0.5,ymax=415.5,
              y dir=reverse,
              height=7cm,width=0.4\textwidth,
              xlabel=Width (Channel),ylabel=Sample (Time)
              ]
              \addplot graphics [xmin=0.5,xmax=24.5,ymin=0.5,ymax=415.5] {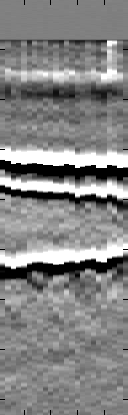};

              \end{axis}
              \node[] at (3.1,4.4) (n0) {};
              \node[rectangle,line width=0,text=red,outer sep=2pt,inner sep=0,callout absolute pointer=(n0),above right=3pt and 20pt of n0] (n0label) {\footnotesize{\textsf{Antenna enclosure}}};
              \draw[draw=red,line width=1.5] (n0.center) -- (n0label.south west);
              
              \node[] at (3.1,3.3) (n1) {};
              \node[rectangle,line width=0,text=red,outer sep=2pt,inner sep=0,callout absolute pointer=(n1),above right=3pt and 20pt of n1] (n1label) {\footnotesize{\textsf{Ground surface}}};
              \draw[draw=red,line width=1.5] (n1.center) -- (n1label.south west);

              \node[] at (3.1,2.9) (n2) {};
              \node[rectangle,line width=0,text=red,outer sep=2pt,inner sep=0,callout absolute pointer=(n2),right=20pt of n2] (n2label) {\footnotesize{\textsf{Subsurface layer boundary}}};
              \draw[draw=red,line width=1.5] (n2.center) -- (n2label.west);

              \node[] at (3.1,2.1) (n3) {};
              \node[rectangle,line width=0,text=red,outer sep=2pt,inner sep=0,callout absolute pointer=(n3),below right=3pt and 20pt of n3] (n3label) {\footnotesize{\textsf{Subsurface layer boundary}}};
              \draw[draw=red,line width=1.5] (n3.center) -- (n3label.north west);
            \end{tikzpicture}
          } 
        \caption{An example grayscale ground-penetrating radar (GPR) image. Four boundary-like features of interest within the image are labeled that we would like to automatically detect, including the reflections due to the ground surface and two subsurface soil layer boundaries.}
        \label{fig:radar_example}
\end{figure}
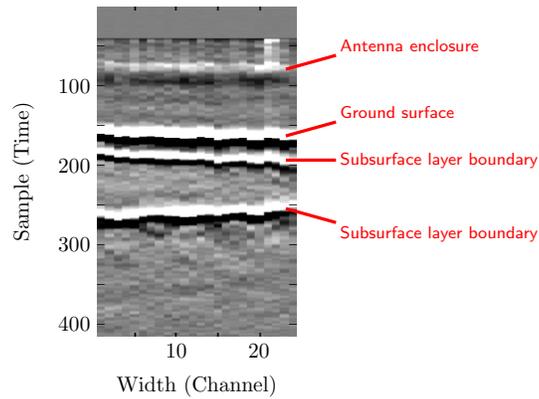

However, layer-boundaries can frequently mimic or obscure landmine signatures in radar images.
An example GPR image, which is created by spacing multiple receivers in a straight line and roughly depicts a two-dimensional cross-section of the ground, is shown in Figure \ref{fig:radar_example}.
In this image, color represents signal amplitude, with white corresponding to positive amplitude, and black corresponding to negative amplitude, and each column in the image is the return signal recorded at a single antenna.
Each strong reflection in the return signal shows up as a \emph{bounce}---a strong positive deviation from equilibrium, followed by a strong negative deviation from equilibrium (or vice versa).

While buried landmines tend to stand out from the soil and cause significant reflections, other large sources of contrasting materials can cause these, as well, such as different soil layers within the ground.
The wide extent of the layers causes the reflections at their boundaries to be primarily depicted as parallel white and black horizontal bands, as shown in Figure \ref{fig:radar_example}.
In individual images, this interference may be difficult to filter out from that of a target signature.
But because soil layering is a phenomenon likely to persist over many radar images of the ground taken in succession, layer-boundaries have the potential to be tracked and their interference mitigated.

\subsubsection{Identifying layer-boundaries using reciprocal pointer chains}

In this section, we present a novel and efficient method for identifying layer boundaries in GPR images.
Efficient methods for processing GPR are particularly important to landmine detection, as the detection occurs in real-time and is time-critical.
Our method builds on an earlier approach to accurately detecting just the air-ground boundary in GPR images \cite{wood2011comparison,smock2011dynamax+}.
The earlier approach represents each pixel in a radar image as a node in a trellis graph.
There is a directed edge between each node/pixel $u$, in column $i$, and each node/pixel $v$, in column $i+1$.
Each edge between neighboring nodes is assigned as a weight a probability score that accounts for the probability of a layer-boundary between the nodes given the local evidence.
This score is maximal when there is both a strong GPR reflection present and when the edge is oriented close to the expected orientation of the ground, which with no prior evidence would be completely horizontal.
The full air-ground boundary is identified as the most probable (most strongly weighted) path across the trellis graph using the Viterbi algorithm.

To find layer boundaries beyond just the air-ground boundary, we use the same graph structure and apply the RPC method to identify locally optimal paths.
However, instead of using the CVPs to identify layer boundaries, we associate the layer boundaries with the RPCs specifically.
There are three notable physical features of layer boundaries that make RPCs particularly suited for identifying them.
First of all, like RPCs, layer boundaries are physically disjoint.
And while the boundaries between different layers may abut, they cannot cross each other.
Likewise, it has been theoretically proven that in the model used here, the use of a Gaussian distribution to describe the orientation probabilities for the edges prevents the RPCs from physically crossing each other \cite{smock2012efficient}.
This is an interesting theoretical result that is specific to the graph model used here.
Thirdly, layer-boundaries can be any width inside of the radar image.
Similarly, RPCs can be anywhere from the full length of the trellis or as short as a single node.
This is in contrast with CVPs, which must span the full length of the trellis.
Thus it is RPCs, and not CVPs, which are most useful here for identifying layer-boundary features.

\begin{figure}[t]
    \centering
    \resizebox{0.8\textwidth}{!}{
           \begin{tabular}{@{}C{0.32\textwidth}C{0.32\textwidth}C{0.32\textwidth}@{}} 
            \begin{tikzpicture}
             \begin{axis}[
              xmin=0.5,xmax=24.5,ymin=0.5,ymax=415.5,
              y dir=reverse,
              height=6cm,width=0.31\textwidth,
              xlabel=Width (Channel),ylabel=Sample (Time)
              ]
              \addplot graphics [xmin=0.5,xmax=24.5,ymin=0.5,ymax=415.5] {gpr_figure_8985.png};
              \end{axis}
            \end{tikzpicture}
&
            \begin{tikzpicture}
             \begin{axis}[
              xmin=0.5,xmax=24.5,ymin=0.5,ymax=415.5,
              y dir=reverse,
              height=6cm,width=0.31\textwidth,
              xlabel=Width (Channel),ylabel=Sample (Time)
              ]
              \addplot graphics [xmin=0.5,xmax=24.5,ymin=0.5,ymax=415.5] {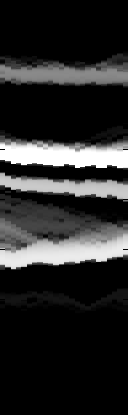};
                  \foreach \y in {1,...,20} {
                  	\addplot[color=white,mark=o,mark size=0.8,line width=0.8]
    			table[header=false,x index=0,y index=\y] {all_RPCs_8985.txt};
			\addplot[color=black,mark=o,mark size=0.35,line width=0.45]			
			table[header=false,x index=0,y index=\y] {all_RPCs_8985.txt};
    		}
              \end{axis}
            \end{tikzpicture}
&
            \begin{tikzpicture}
             \begin{axis}[
              xmin=0.5,xmax=24.5,ymin=0.5,ymax=415.5,
              y dir=reverse,
              height=6cm,width=0.31\textwidth,
              xlabel=Width (Channel),ylabel=Sample (Time)
              ]
              \addplot graphics [xmin=0.5,xmax=24.5,ymin=0.5,ymax=415.5] {gpr_figure_8985.png};
                  \foreach \y in {1,...,4} {
			\addplot[color=black,line width=1.6]
    			table[header=false,x index=0,y index=\y] {layer_boundaries_8985.txt};
			\addplot[color=orange,line width=1.4]
    			table[header=false,x index=0,y index=\y] {layer_boundaries_8985.txt};
    		}
                \end{axis}
            \end{tikzpicture}
\\	
               (a) GPR & (b) Inverse cost map and reciprocal pointer chains & (c) Layer boundaries
	 \end{tabular}
        } 
         \caption{An example of subsurface layer-boundary identification in GPR data using RPCs. These results are discussed in detail in Section \ref{sec:layer_results}.}
	\label{fig:example_8985}     
\end{figure}
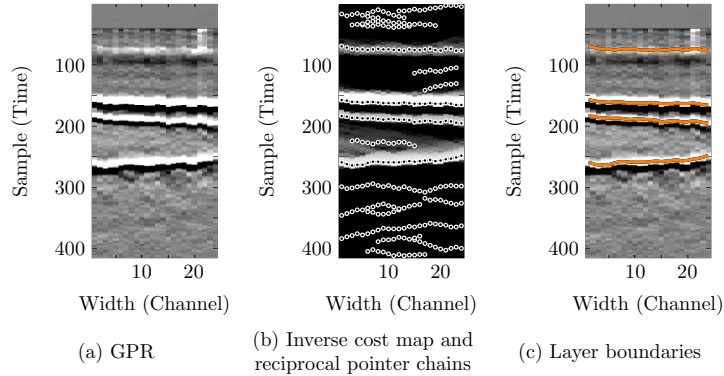

\begin{figure}[t]
    \centering
        \resizebox{0.8\textwidth}{!}{
           \begin{tabular}{@{}C{0.32\textwidth}C{0.32\textwidth}C{0.32\textwidth}@{}} 
            \centering
            \begin{tikzpicture}
             \begin{axis}[
              xmin=0.5,xmax=24.5,ymin=0.5,ymax=415.5,
              y dir=reverse,
              height=6cm,width=0.31\textwidth,
              xlabel=Width (Channel),ylabel=Sample (Time)
              ]
              \addplot graphics [xmin=0.5,xmax=24.5,ymin=0.5,ymax=415.5] {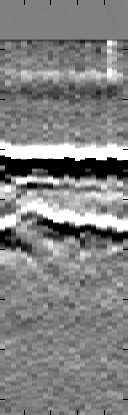};
              \end{axis}
            \end{tikzpicture}
&
            \begin{tikzpicture}
             \begin{axis}[
              xmin=0.5,xmax=24.5,ymin=0.5,ymax=415.5,
              y dir=reverse,
              height=6cm,width=0.31\textwidth,
              xlabel=Width (Channel),ylabel=Sample (Time)
              ]
              \addplot graphics [xmin=0.5,xmax=24.5,ymin=0.5,ymax=415.5] {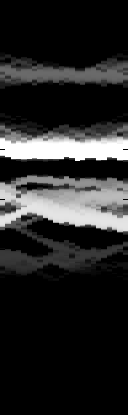};
                  \foreach \y in {1,...,20} {
                  	\addplot[color=white,mark=o,mark size=0.8,line width=0.8]
    			table[header=false,x index=0,y index=\y] {all_RPCs_4818.txt};
			\addplot[color=black,mark=o,mark size=0.35,line width=0.45]			
			table[header=false,x index=0,y index=\y] {all_RPCs_4818.txt};
    		}
              \end{axis}
            \end{tikzpicture}
&
            \begin{tikzpicture}
             \begin{axis}[
              xmin=0.5,xmax=24.5,ymin=0.5,ymax=415.5,
              y dir=reverse,
              height=6cm,width=0.31\textwidth,
              xlabel=Width (Channel),ylabel=Sample (Time)
              ]
              \addplot graphics [xmin=0.5,xmax=24.5,ymin=0.5,ymax=415.5] {gpr_figure_4818.png};
                  \foreach \y in {1,...,5} {
			\addplot[color=black,line width=1.6]
    			table[header=false,x index=0,y index=\y] {layer_boundaries_4818.txt};
			\addplot[color=orange,line width=1.4]
    			table[header=false,x index=0,y index=\y] {layer_boundaries_4818.txt};
    		}
                \end{axis}
            \end{tikzpicture}
\\	
               (a) GPR & (b) Inverse cost map and reciprocal pointer chains & (c) Layer boundaries
  	 \end{tabular}
        } 
         \caption{A second example of subsurface layer-boundary identification in GPR data using RPCs, also discussed in detail in Section \ref{sec:layer_results}.}
	\label{fig:example_4818}     
\end{figure}
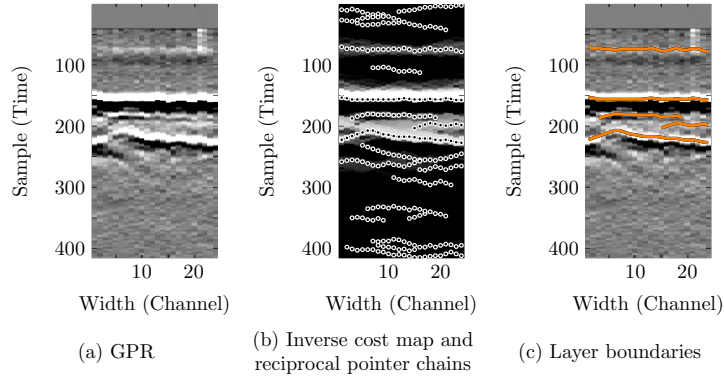

Figures \ref{fig:example_8985} and \ref{fig:example_4818} show qualitative results of the approach on two different GPR images.
Figures \ref{fig:example_8985}.a and \ref{fig:example_4818}.a show the original GPR images.
Figures \ref{fig:example_8985}.b and \ref{fig:example_4818}.b show the RPCs computed for the radar images.
These are overlaid on top of a plot of the via-path probability (inverse cost) of every pixel, where white indicates higher probability and black indicates lower probability.
Plotting the via-path probability for every node shows the plateaus in the graph, which clearly coincide with the RPCs in this case.
This also helps to illustrate the local optimality of the paths.
Figures \ref{fig:example_8985}.c and \ref{fig:example_4818}.c, meanwhile, show the final result with the identified layer boundaries in orange.
To produce the final set of layer boundaries, we ranked the RPCs and thresholded them based on path cost.

A qualitative evaluation of the results shows this method to be useful at identifying layer boundaries and locating their positions with a relatively high-degree of accuracy.
Overall, this demonstrates that RPCs can be an effective strategy to agglomerate local evidence, and their disjointness, efficient computation, and variable length proved to be particularly well-matched to the problem of layer-boundary identification in GPR.

\section{Conclusion}\label{sec:conclusion}

In this work, we have presented the first comprehensive formalization of cascading via-paths (CVPs) and the reciprocal pointer chain (RPC) method from a graph-theoretic point of view.
Prior work in this area was application-oriented and lacked both a satisfying theoretical motivation and the abstraction necessary for the method to be applicable beyond just specific types of graphs.
This work is not only the first attempt to consolidate the theory and standardize the terminology surrounding the approach, but also greatly expands upon the theory to open up new avenues for theoretical work and practical application.

Our theoretical contributions are both mathematical and computational in nature.
Among these, we showed that due to hidden assumptions about the type of graph used in prior applications, a previously unknown difference existed between two earlier proposed strategies---the RPC method and the disjoint plateau method---for enumerating via-paths given two shortest path trees.
We sorted out this difference, distinguishing CVPs from the set of all possible via-paths and establishing the RPC method as the only correct procedure for CVP enumeration in all types of graphs.

In addition, we proved the key result that the characteristic RPCs of any two CVPs must be disjoint.
From this, we derived a number of properties of CVPs that can be computed in $O(|E| + |V| \log |V|)$ time.
These theorems enable a versatile set of new and efficient path computation approaches.
For example, we derived a lower bound on the similarity between a CVP and any shorter path between $s$ and $t$, which does not require any comparisons between paths.
Among other things, this allows us to compute a diverse set of paths without having to consider a combinatorial optimization over all sets of paths, which is important in applications concerned with efficient candidate-set generation.
Additionally, we used the disjointness property to propose, using RPCs, a potentially significantly faster approach to computing the $k$ shortest loop-less paths in a graph.

Finally, we presented two applications of CVPs and RPCs---alternative route finding in road networks and layer-boundary identification in ground-penetrating radar (GPR) images---along with qualitative results that demonstrate the RPC method is effective at finding locally optimal solutions to path problems.
While this has been demonstrated to an extent previously, part of our motivation has been to provide a much more general foundation and principled justification for the effectiveness of these methods.
Additionally, we demonstrated a completely novel use of the RPC method with the RPCs themselves as the end product for identifying layer-boundaries, due to their ability to agglomerate local evidence for linear features in a radar image.
Overall, our hope is that by demonstrating the versatility, depth, and practicality of the method and supporting theory, we are able to significantly broaden their perceived applicability and stimulate new research in this area, both applied and theoretical.


\section*{Acknowledgements}

This work was supported in part by the Army Research Office under Grant W911NF-08-10410. The views and conclusions contained in this document are those of the authors and should not be interpreted as representing the official policies either expressed or implied, of the Army Research Office, Army Research Laboratory, or the U.S. Government.  The U.S. Government is authorized to reproduce and distribute reprints for Government purposes notwithstanding any copyright notation hereon. The authors would like to thank L. Dai, R. Weaver, P. Howard, and T. Donzelli for their support of this work.


\clearpage

\bibliography{refs}

\begin{thebibliography}{10}

\bibitem{abraham2013alternative}
I.~Abraham, D.~Delling, A.~V. Goldberg, and R.~F. Werneck.
\newblock Alternative routes in road networks.
\newblock {\em Journal of Experimental Algorithmics (JEA)}, 18:1--3, 2013.

\bibitem{cormen2001introduction}
T.~H. Cormen, C.~E. Leiserson, R.~L. Rivest, and C.~Stein.
\newblock {\em Introduction to algorithms}, volume~2.
\newblock MIT press Cambridge, 2001.

\bibitem{demetrescu20069th}
C.~Demetrescu, A.~Goldberg, and D.~Johnson.
\newblock 9th {DIMACS} implementation challenge---shortest paths.
\newblock {\em American Mathematical Society}, 2006.

\bibitem{dijkstra1959note}
E.~W. Dijkstra.
\newblock A note on two problems in connexion with graphs.
\newblock {\em Numerische Mathematik}, 1(1):269--271, 1959.

\bibitem{fredman1987fibonacci}
M.~L. Fredman and R.~E. Tarjan.
\newblock Fibonacci heaps and their uses in improved network optimization
  algorithms.
\newblock {\em Journal of the ACM (JACM)}, 34(3):596--615, 1987.

\bibitem{fujita2003dual}
Y.~Fujita, Y.~Nakamura, and Z.~Shiller.
\newblock Dual {D}ijkstra search for paths with different topologies.
\newblock In {\em Robotics and Automation, 2003. Proceedings. ICRA'03. IEEE
  International Conference on}, volume~3, pages 3359--3364. IEEE, 2003.

\bibitem{jimenez2003lazy}
V.~M. Jim{\'e}nez and A.~Marzal.
\newblock A lazy version of eppsteinÕs k shortest paths algorithm.
\newblock In {\em Experimental and efficient algorithms}, pages 179--191.
  Springer, 2003.

\bibitem{jones2012method}
A.~H. Jones.
\newblock Method of and apparatus for generating routes, Aug.~21 2012.
\newblock US Patent 8,249,810.

\bibitem{jones2012method2}
A.~H. Jones.
\newblock Method of and apparatus for generating routes, Aug.~18 2012.
\newblock European Patent 2,084,493.

\bibitem{li2013optimization}
Y.~Li, Q.~Yang, W.~Sima, J.~Li, and T.~Yuan.
\newblock Optimization of transmission-line route based on lightning incidence
  reported by the lightning location system.
\newblock {\em Power Delivery, IEEE Transactions on}, 28(3):1460--1468, 2013.

\bibitem{lombard1993gateway}
K.~Lombard and R.~Church.
\newblock The gateway shortest path problem: generating alternative routes for
  a corridor location problem.
\newblock {\em Journal of Geographical Systems}, 1(1):25--45, 1993.

\bibitem{luxen2012candidate}
D.~Luxen and D.~Schieferdecker.
\newblock Candidate sets for alternative routes in road networks.
\newblock In {\em Experimental Algorithms}, pages 260--270. Springer, 2012.

\bibitem{nakamura2003dual}
Y.~Nakamura and Y.~Fujita.
\newblock Dual dijkstra search for planning multipe paths, Feb.~7 2003.
\newblock US Patent App. 10/359,700.

\bibitem{nicholson1966finding}
T.~A.~J. Nicholson.
\newblock Finding the shortest route between two points in a network.
\newblock {\em The Computer Journal}, 9(3):275--280, 1966.

\bibitem{pinto2009beyond}
N.~Pinto and T.~H. Keitt.
\newblock Beyond the least-cost path: evaluating corridor redundancy using a
  graph-theoretic approach.
\newblock {\em Landscape Ecology}, 24(2):253--266, 2009.

\bibitem{scaparra2014corridor}
M.~P. Scaparra, R.~L. Church, and F.~A. Medrano.
\newblock Corridor location: the multi-gateway shortest path model.
\newblock {\em Journal of Geographical Systems}, pages 1--23, 2014.

\bibitem{shiller2004computing}
Z.~Shiller, Y.~Fujita, D.~Ophir, and Y.~Nakamura.
\newblock Computing a set of local optimal paths through cluttered environments
  and over open terrain.
\newblock In {\em Robotics and Automation, 2004. Proceedings. ICRA'04. 2004
  IEEE International Conference on}, volume~5, pages 4759--4764. IEEE, 2004.

\bibitem{smock2011dynamax+}
B.~Smock, P.~Gader, and J.~Wilson.
\newblock Dyna{M}ax+ ground-tracking algorithm.
\newblock In {\em Proceedings of SPIE Detection and Sensing of Mines, Explosive
  Objects, and Obscured Targets XVI}, volume 8017, page 80171I. International
  Society for Optics and Photonics, 2011.

\bibitem{smock2012efficient}
B.~Smock and J.~Wilson.
\newblock Efficient multiple layer boundary detection in ground-penetrating
  radar data using an extended viterbi algorithm.
\newblock In {\em SPIE Defense, Security, and Sensing}, pages 83571X--83571X.
  International Society for Optics and Photonics, 2012.

\bibitem{smock2012reciprocal}
B.~Smock and J.~Wilson.
\newblock Reciprocal pointer chains for identifying layer boundaries in
  ground-penetrating radar data.
\newblock In {\em Geoscience and Remote Sensing Symposium (IGARSS), 2012 IEEE
  International}, pages 602--605. IEEE, 2012.

\bibitem{suurballe1974disjoint}
J.~Suurballe.
\newblock Disjoint paths in a network.
\newblock {\em Networks}, 4(2):125--145, 1974.

\bibitem{camvit2006choice}
C.~{V}ehicle {I}nformation~{T}echnology {L}td.
\newblock Choice routing explanation.
\newblock
  \url{http://www.camvit.com/camvit-technical-english/Camvit-Choice-Routing-Explanation-english.pdf},
  2006.

\bibitem{wood2011comparison}
J.~Wood, J.~Bolton, G.~Casella, L.~Collins, P.~Gader, T.~Glenn, J.~Ho, W.~Lee,
  R.~Mueller, B.~Smock, et~al.
\newblock Comparison of algorithms for finding the air-ground interface in
  ground penetrating radar signals.
\newblock In {\em Proceedings of SPIE Detection and Sensing of Mines, Explosive
  Objects, and Obscured Targets XVI}, volume 8017, page 80171L. International
  Society for Optics and Photonics, 2011.

\bibitem{yen1971finding}
J.~Y. Yen.
\newblock Finding the k shortest loopless paths in a network.
\newblock {\em Management Science}, 17(11):712--716, 1971.

\end{thebibliography}
\bibliographystyle{abbrv}

\end{document}